\documentclass[12pt]{article}
\usepackage{indentfirst}
\usepackage{mathpazo}
\usepackage{setspace}
\usepackage{rotating}
\usepackage{bbm}
\usepackage{graphicx}
\usepackage{tikz}
\usetikzlibrary{decorations.pathreplacing,calligraphy}
\usepackage{rotating}
\usepackage{amsmath,amsthm}
\usepackage{amssymb}
\usepackage{booktabs} 
\usepackage{fullpage}
\usepackage{multicol}
\usepackage{lettrine}
\usepackage{bm}
\usepackage{algorithm}
\usepackage{algorithmic}
\usepackage{subfig}

\usepackage{tablefootnote}

\usepackage{float}
\usepackage{threeparttable}
\newtheorem{assumption}{Assumption}

\newtheorem{proposition}{Proposition}
\newtheorem{theorem}{Theorem}
\theoremstyle{definition}

\theoremstyle{example}

\newtheorem{remark}{Remark}

\usepackage[colorlinks=true,
citecolor=blue,
linkcolor=blue,
anchorcolor=blue,
urlcolor=blue]{hyperref}

\newcommand\independent{\protect\mathpalette{\protect\independenT}{\perp}}
\def\independenT#1#2{\mathrel{\rlap{$#1#2$}\mkern2mu{#1#2}}}

\usepackage[style=authoryear-comp,natbib=true,maxnames=2, backend=bibtex]{biblatex}

\bibliography{literature}

\usepackage[titletoc,toc,title,page,header]{appendix}
\usepackage{minitoc}

\noptcrule
\usepackage{caption}
\usepackage{subcaption}

\title{Extracting Mechanisms from Heterogeneous Effects: An Identification Strategy for Mediation Analysis \footnote{I thank Nathaniel Beck, Adam N. Glynn, Donald Green, Kosuke Imai, Dimitri Landa, Connor Jerzak, John Marshall, Cyrus Samii, Yuki Shiraito, Tara Slough, Anna Wilke, Teppei Yamamoto; seminar audiences at Columbia University, New York University, and University of North Carolina at Chapel Hill; and participants at PolMeth XL and APSA 2023 for helpful feedback.}}

\date{\today}

\author{Jiawei Fu\footnote{Assistant Professor, Duke University \url{jiawei.fu@duke.edu}}}

\begin{document}
\maketitle
\singlespacing

\begin{abstract}
Understanding causal mechanisms is crucial for explaining and generalizing empirical phenomena. Causal mediation analysis offers statistical techniques to quantify the mediation effects. Although numerous methods have been developed for causal inference more broadly, the methodological toolkit for causal mediation analysis remains limited. Current methods often require multiple ignorability assumptions or sophisticated research designs. In this paper, we introduce an alternative identification strategy that enables the simultaneous identification and estimation of treatment and mediation effects. By combining explicit and implicit mediation analysis, this strategy leverages heterogeneous treatment effects and does not require addressing some unobserved confounders. Monte Carlo simulations demonstrate that the method is more accurate and precise across various scenarios. To illustrate the efficiency and efficacy of our method, we apply it to estimate the causal mediation effects in two studies with distinct data structures, focusing on common pool resource governance and voting information. 

    \noindent 

\vspace{.1in}



\end{abstract}

\thispagestyle{empty}
\doublespacing
\clearpage

\doparttoc 
\faketableofcontents 

\setcounter{page}{1}

\section{Introduction}

“How and why does the treatment affect the outcome?” This causal mechanism question lies at the core of social science research, guiding both theoretical development and the generalization of empirical findings. Causal mediation analysis answers this question by quantifying the portion of a treatment’s effect that operates through intermediate variables between the treatment and the outcome \citep{vanderweele2015explanation}. The dominant approach in social science relies on the assumption of sequential ignorability \citep{imai2011unpacking}. Although numerous methods have been developed for causal inference more broadly, the methodological toolkit for causal mediation analysis remains limited.\footnote{Only around $10\%$ of published papers have the opportunity to employ formal causal mediation analysis, according to \citet{blackwell2024assumption}.} This study proposes an alternative identification strategy that explicitly quantifies causal mediation effects by leveraging heterogeneous treatment effects (HTEs).

Mediation analysis has been approached through two main methods. In \textit{explicit mediation analysis}, researchers aim to identify the exact causal mediation effect—that is, the causal effects mediated by a specific proposed mediator. Using either the counterfactual or structural framework, as illustrated in Figure \ref{fig:illus}, the total treatment effect is decomposed into direct and indirect effects. Because of the existence of a mediator in the indirect effect, mediation analysis requires stronger assumptions than what is required for the identification of total treatment. The identification of treatment effects requires addressing confounding between the treatment and the outcome ($U_1$ and $U_3$ in the right panel), while mediation analysis additionally requires the ignorability of the mediator ($U_2$, $U_4$, and $U_1$) \citep{imai2010identification}. In some empirical studies, satisfying and justifying the second ignorability assumptions is challenging, even in the traditional randomized experiment.


\begin{figure}
   \begin{center}
       \begin{tikzpicture}
    
            \node at (8,0) (t1) {$Treatment$};
           \node at (10.5,0) (m1) {$Mediator$};
           \node at (13,0) (y1) {$Outcome$}; 
           \node at (10.5,-1) (y) {Indirect Effect};
           
           \node[rotate=-90] at (14.5,0.6) (ta) {Overall Effect}; 
           \path[->] (t1) edge node[auto] {} (m1);
          \path[->] (m1) edge node[auto] {} (y1);
          \path[->]
            (t1) edge [bend left=45] node[auto] {Direct Effect} (y1);
             \draw [pen colour={orange}, decorate,
    decoration = {calligraphic brace,mirror,amplitude=5pt}] (8,-0.3) --  (13,-0.3);

            \draw [pen colour={orange}, decorate,
    decoration = {calligraphic brace,amplitude=5pt}] (14,1.5) --  (14,-0.4);

        \node at (16,0) (t) {$T$};
           \node at (18,0) (m) {$M$};
           \node at (18,0.8) (m') {$U_4$};
           \node at (20,0) (y) {$Y$};
           \node at (17,-1) (u1) {$U_1$};
           \node at (19,-1) (u2) {$U_2$};
           \node at (18,-2) (u3) {$U_3$};
           \draw[->] (t) -- (m);
           \draw[->] (m) -- (y);
           \draw[->] (u1) -- (t);
           \draw[->] (u1) -- (m);
           \draw[->] (u2) -- (m);
           \draw[->] (u2) -- (y);
           \draw[->] (t) -- (m');
          \draw[->] (m') -- (y);
          \draw[->] (m') -- (m);
           \draw[->]
            (u3) edge [bend left=45] (t)
            (u3) edge [bend right=45] (y);
            
          \draw[->] (t) .. controls (16,1.44)and (17.2,2) .. (18,2)
               .. controls (18.9,2) and (20,1.44) .. (y);  
       \end{tikzpicture}
   \end{center}
    \caption{The left panel illustrates the basic decomposition of the overall treatment effect. The right panel is one example of the directed acyclic graph. $T$ is the treatment variable, $M$ is the mediator, $Y$ is the outcome variable, and $U_j, j=1,2,3,4$ are confounders.}\label{fig:illus}
\end{figure}

In \textit{implicit mediation analysis}, researchers do not seek to identify the exact causal mediation effect but rather aim to derive qualitative insights into the underlying mechanism \citep{bullock2021failings}. For example, they may estimate the correlation between the treatment and mediator or between the mediator and the outcome \citep{blackwell2024assumption}. The most popular approach utilizes HTEs. The underlying intuition is that if a specific mechanism is active for certain units, it will generate HTEs with respect to a particular covariate. If the mechanism is inert, it will fail to produce HTEs \citep{fu2023heterogeneous}. This method is widely used because it does not require measuring the mediator or relying on strong identification assumptions.

Our identification strategy incorporates the strength of two approaches and bypass some ignobility assumptions. We first decompose the total treatment effect by highlighting the treatment effect on the mediator. As we show in section \ref{sec:model}, the new decomposition has the form of simple linear regression, where the `dependent variable' is the average treatment effect on the outcome and the `independent variable' is the average total treatment effect on the mediator. To \textit{non-parametrically} identify the causal mediation effects (i.e. the slope), we assume that the average treatment effect on the mediator is not correlated with the error term. Under this assumption, even if confounders exist between the mediator and the outcome (so that sequential ignorability fails), we can still identify the causal mediation effect. In the instrumental variables literature, a related identification assumption appears in \citet{kolesar2015identification} and in the Mendelian randomization framework \citep{bowden2015mendelian}. Actually, \citet{fu2023heterogeneous} shows that this kind of isolated mechanism assumption is also generically necessary for HTEs to reflect \textit{qualitative} information about mechanism activation, even if not in a quantitative sense. In this respect, our approach does not impose additional requirements; if this assumption fails to hold for all covariates, HTEs cannot be used to study mechanism activation in the usual way. 

To implement the strategy, we rely on HTEs on the mediator as our primary data sources. If the same research is conducted several times in different locations, we naturally obtain different treatment effects from those studies. In section \ref{sec:des}, we propose two more research designs to exploit those HTEs. In the first design, we suggest researchers use pre-treatment covariates to identify those subgroups, probably with data-driven methods like causal tree or forest \citep{wager2018estimation}. Next, for each subgroup, we obtain the required data (average treatment effects on the outcome and on the mediator), which we refer to as \textit{Heterogeneous Subgroup Design}. It is worth emphasizing that our goal is not to identify the true underlying subgroups, but rather to obtain sufficient variation in the data. The second design explores variation in treatment intensity. For example, in a Get-Out-The-Vote (GOTV) experiment, researchers might randomly assign different numbers of canvassing mails. By manipulating the magnitude or intensity of the treatment across different arms, we can induce varying treatment effects on the mediator, a design we refer to as the \textit{Multiple Treatment Meta Design}. 

The estimation and inference are discussed in section \ref{sec:infer}. We then use Monte Carlo simulation to examine the performance of those estimators in section \ref{sec:appli}. Our method demonstrates greater robustness and efficiency. We apply our methodology to estimate causal mediation effects in two empirical studies using actual data. The first study is about `Governance on Resources'. We demonstrate the estimation of mediation effects using just six group-level estimates derived from six different sites. The second examines the impact of information on voting behavior, where we calculate the mediation effect using individual-level data. We quantify the causal mediation effect of “bad/good information” about legislators on voting choice through voters’ perceptions of politicians’ effort. In section \ref{sec:ext}, we offer a range of valuable extensions designed to assist applied researchers in meeting our core identification assumptions through both the design stage and the data analysis stage.

In general, we introduce an alternative identification strategy for causal mediation analysis. Causal mediation analysis is challenging, and each method relies on its own identification assumptions. We hope our method provides an additional useful tool for applied researchers who are eager to learn about causal mechanisms. Instead of relying on unconfoundedness assumptions, our method leverages HTEs and the theoretical structure of mechanisms. \textit{No single method is suitable for all study types}: identification based on "exogeneity" or "ignorability" should be applied when researchers have sufficient covariates and understanding of potential confounders. Conversely, when confounders are unclear or unmeasured, our approach—relying on HTEs and knowledge of mechanisms—may be an effective alternative strategy.

\section{Causal Mediation Framework}\label{sec:decom}

In this section, we review basic definitions and notations for mediation analysis, if necessary, with the help of directed acyclic graphs (DAG).\footnote{Although there exist some debates about the potential outcome approach and graphic approach (see \citet{imbens2020potential}, and Pearl's reply: http://causality.cs.ucla.edu/blog/index.php/2020/01/29/on-imbens-comparison-of-two-approaches-to-empirical-economics/), we still find both approaches have their own particular merits.} 

\subsection{Running Example}

Throughout this paper, we will use the Get-Out-The-Vote experiments as running examples to illustrate various concepts.  Consider the scenario where researchers have observed that treatments such as door-to-door canvassing, phone calls, or mailings can significantly increase voter turnout \citep{green2019get}. A critical question arises: through what mechanisms, do these treatments influence turnout? Mediation analysis helps us quantify the average causal effects along these pathways.

To better understand the mechanisms involved in this example, we draw on political economy theory. According to \citet{riker1968theory}, there are four elements that influence a voter’s decision to abstain or vote, as shown in Figure \ref{fig:voting}. (1) Cost of Voting ($C$): This includes costs related to information gathering, registration, and transportation, among others.  Generally, higher costs discourage turnout. (2) Civic Duty ($D$): In democracies, voting is often portrayed as a civic responsibility. $D$ captures the utility derived from fulfilling this civic duty; fail to fulfill the duty may lead to social pressure as dis-utility. (3) Utility Difference ($B$): Voters are more likely to turn out if the utility gained from their preferred candidate winning significantly outweighs the utility loss if they lose. (4) Pivotality ($p_i$): This refers to the probability that a voter’s vote could decisively impact the election outcome, with $p_i \in [0,1]$ representing the probability. A lower likelihood of affecting the result may decrease the motivation to vote. 

\begin{figure}[!h]
    \centering
 \includegraphics[width=0.8\linewidth]{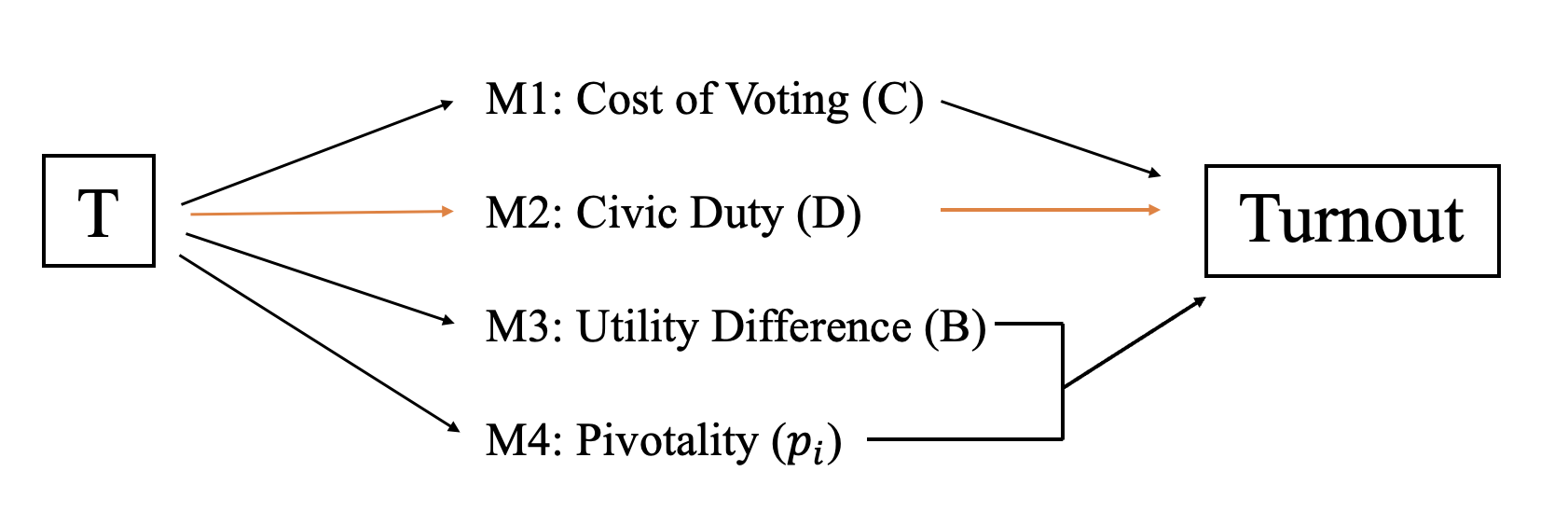}
    \caption{Mechanisms of Turnout Derived from Theory.}
    \label{fig:voting}
\end{figure}

Consequently, a voter prefers to vote over abstaining if and only if $p_i B + D - C \ge 0$; that is, the utility from voting outweighs the costs. These four elements—cost, civic duty, utility difference, and pivotality—represent the mechanisms that can influence turnout. As the formula clearly demonstrates, Mechanisms 3 and 4 are correlated. In our causal mediation question, we specifically focus on the civic duty mechanism, measured by social pressure, to quantify its causal effect on voter turnout.

\subsection{Causal Decomposition}

In the counterfactual approach of causal inference, causation and mediation are interpreted and decomposed with potential outcomes \citep{holland1986statistics}. Let $T$ be the binary treatment and $Y$ be the outcome variable. The overall treatment effect of $T$ on $Y$ for individual $i$, denoted by $\tau^i$, is represented by
$$
\tau^i = Y^i(1,M^i_1(1))- Y^i(0,M^i_1(0)).
$$ For the total effects, mediators $M$ should consider potential outcomes under the treatment status. For instance, $M_1^i(1)$ represents the potential social pressure when the treatment $T=1$. 



Total treatment effects can be decomposed into natural direct and indirect effects \citep{robins1992identifiability}. We call $\delta^i(t)= Y^i(1,M^i(t))-Y^i(0,M^i(t))$ the natural direct effect for $t=0, 1$, where the mediator is set to the value it would have been under treatment $t$. It captures the effects that are not transmitted by the mediator of interest. \footnote{The terminology ``natural'' is in contrast to ``controlled.'' For controlled direct effect, we fix the mediator at a certain value $m$, rather than their potential outcomes under a given treatment assignment. Therefore, the controlled direct effect can be defined as $Y^i(1, m)- Y^i(0, m)$ \citep[see][]{acharya2016explaining}.} Similarly, $\eta^i(t) = Y^i(t,M^i(1))-Y^i(t,M^i(0))$ is the natural indirect effect for $t=0, 1$. It denotes the treatment effect through the mediating variable. In our example, this refers to how the decision to turn out changes in response to the shift in social pressure from what it would be under the control condition ($M^i(0)$) to the treatment condition ($M^i(1)$), while holding the treatment status constant at $t$. \footnote{See more discussions in supplementary materials (SI) \ref{si:decomp}.}

Due to the fundamental problem of causal inference, we are confined to the estimation of aggregate effects, typically averages. We therefore define $\tau$ as the average overall treatment effect: $\tau := \mathbb{E}[\tau^i]$,  $\delta:=\mathbb{E}[\delta^i]$ as the average direct effect, and $\eta:= \mathbb{E} [\eta^i]$ as the average indirect effect for mechanism represented by the mediator $M$. As is standard, illustrated in the left DAG in Figure \ref{fig:dag1}, we can then decompose the total causal effect into the sum of natural direct and indirect effects: $\tau = \delta(t)+\eta(1-t)$.

\begin{figure}
   \begin{center}
       \begin{tikzpicture}
           \node at (0,0) (t) {$T$};
           \node at (2,0) (m) {$M$};
           \node at (4,0) (y) {$Y$};   
           \node at (2,-0.7) (x) {$\eta$}; 
          \path[->] (t) edge  (m);
          \path[->] (m) edge  (y);
           \path[->]
            (t) edge [bend left=45] node[auto] {$\delta$} (y);   
            \draw [pen colour={orange}, decorate,
    decoration = {calligraphic brace,mirror,amplitude=5pt}] (0,-0.3) --  (4,-0.3);

            \node at (6,0) (t1) {$T$};
           \node at (8,0) (m1) {$M$};
           \node at (10,0) (y1) {$Y$}; 
           \node at (10.5,0.6) (ta) {$\tau$}; 
           \path[->] (t1) edge node[auto] {$\gamma$} (m1);
          \path[->] (m1) edge node[auto] {$\tilde{\beta}$} (y1);
          \path[->]
            (t1) edge [bend left=45] node[auto] {$\delta$} (y1);

            \draw [pen colour={orange}, decorate,
    decoration = {calligraphic brace,amplitude=5pt}] (10.2,1.5) --  (10.2,-0.4);
          
       \end{tikzpicture}
   \end{center}
   \caption{Decompositions with Counterfactual and Structural Approaches. $\delta$ is the direct effect, $\eta$ is the indirect effect, $\tau$ is the total effect, $\gamma$ is the treatment effect on the mediator, $\tilde{\beta}$ is the mediator effect on the outcome.}\label{fig:dag1}
\end{figure}

Historically, path analysis and structural equation modeling (SEM) is the predominantly used framework for conducting mediation analysis \citep[see][]{mackinnon2012introduction,hong2015causality}. As a special case, \citet{baron1986moderator} develop the linear additive model using a single treatment, mediator, and outcome variable. It comprises two main equations:
\begin{align}
        Y &= \alpha_1 + \delta T + \tilde{\beta} M + \varepsilon_1 \label{equ:sem1}\\
        M &= \alpha_2 + \gamma T +   \varepsilon_2 \label{equ:sem2}
\end{align}

Replacing $M$ (equation \eqref{equ:sem2}) in \eqref{equ:sem1}, we obtain the reduced form as follows:
\begin{align}
    Y &= (\alpha_1+\alpha_2 \tilde{\beta}) + (\delta+ \tilde{\beta} \gamma)T + (\varepsilon_1+\tilde{\beta} \varepsilon_2) \label{equ:sem3}\\
    &:= \alpha_3 + \tau T+\varepsilon_3 \label{equ:sem4}
\end{align}

Traditionally, parameter before the treatment $T$, i.e., $\tau=\delta+\tilde{\beta} \gamma$ in \eqref{equ:sem4} is interpreted as the total treatment effect; $\delta$ is interpreted as the direct effect; and $\tilde{\beta} \gamma$ is interpreted as the indirect effect. The right part of Figure \ref{fig:dag1} illustrates this DGP. Unlike the counterfactual approach, structural models implicitly incorporate parametric assumptions, such as linearity and constant effects. 


\subsection{Explicit Mediation Analysis}\label{sec:idenassump}

Explicit mediation analysis aims to quantify the exact mediation effects. This essentially requires the sequential ignorability to non-parametrically identify the causal mediation effect. Multiple versions of sequential ignorability assumption exist. One of the most concise versions is given by \citet{imai2010identification}. Formally, it has two important parts. The first part is similar to the unconfoundedness assumption in causal inference. Essentially, it requires treatment assignment to be ignorable given the observed pretreatment confounders $X$: 
\begin{align}
    \{Y^i(t',m), M^i(t)\} \independent T^i | X^i=x \label{equ:seq1}
\end{align}

Notably, the treatment value is different for outcome $Y$ and mediator $M$. Hence, it specifies the full joint distribution of all the potential outcome and mediator variables \citep{ten2012review}. The second part entails the mediator is ignorable given the observed treatment and pre-treatment confounders: 
\begin{align}
    Y^i(t',m) \independent M^i(t)|T^i=t,X^i=x\label{equ:seq2}
\end{align} In the assumption \eqref{equ:seq2}, the mediator takes the value at the ``current'' treatment assignment $t$, but the potential outcome is under treatment assignment $t'$. For example, it requires that the potential social pressure under treatment condition is independent of the potential turnout under control condition, given the individual is under the treatment status and has pre-treatment variables $X^i=x$. This cross-world argument makes it challenging to be satisfied in some studies. 

One important limitation of both assumptions is that all covariates $X$ must be pre-treatment; generally, the natural indirect effect is not identified even if we have data on the post-treatment confounders \citep{avin2005identifiability}. There have been several attempts to relax these constraints by introducing additional assumptions. \citet{robins2003semantics} proposes the finest fully randomized causally interpreted structured tree graph (FRCISTG) model. Under this semantic framework for causal DAGs, the ``cross-world/indices'' property can be relaxed, allowing for post-treatment confounders. However, an additional no-interaction assumption is required to nonparametrically identify the causal mediation effect. 
Similar assumption is also discussed in \citet{glynn2012product}. \citet{rudolph2023efficient} and \citet{tchetgen2014identification} explore how a monotonicity assumption can help identify causal mediation effects when a confounder is affected by treatment. Recently, a new strand of literature has proposed another estimand, the interventional indirect effect, which circumvents these strong assumptions \citep{diaz2021nonparametric,miles2023causal}.


Typically, to non-parametrically identify natural indirect effects, sequential ignorability requires us to account for all pre-treatment confounders affecting treatment, mediator, and outcome. Practically, researchers hardly observe and measure all confounders, and it is challenging to ensure all confounders are under control. Because of the challenge, in practice, researchers often rely on modeling assumptions to estimate the causal mediation effect. The traditional choice is the linear regression model (equation \eqref{equ:sem1}-\eqref{equ:sem4}). Several parametric assumptions are required to identify parameters and thus the indirect effect \citep{mackinnon2012introduction}. In particular, we need two assumptions: 
\begin{enumerate}
    \item Correct function form, which primarily means linear in parameter and additivity;
    \item No omitted variable, especially error terms $\epsilon_j$ should not correlate across equations.
\end{enumerate}
Those function-form assumptions can also be interpreted by counterfactual languages \citep[See][]{jo2008causal,sobel2008identification}. Generally, explicit mediation analysis under the linear regression model still requires several ``exogeneity'' assumptions to identify parameters in regression equations ($\tilde{\beta}$ and $\gamma$ or $\tau$ and $\delta$). As shown in Figure \ref{fig:dag2}, both non-parametric and model-based assumptions require controlling $U_1$ and $U_2$. However, in practice, it is difficult to measure and control all such confounders.

\subsection{Implicit Mediation Analysis}

Given these challenges, researchers frequently rely on implicit mediation analysis \citep{bullock2021failings}. Although it does not necessarily reveal the causal mediation effect, it can provide useful qualitative insights into the mechanism. However, many practices incorporate unstated assumptions. 

The dominant approach to implicit mediation analysis is to examine HTEs to assess whether a proposed mechanism is active. Although mediation is conceptually distinct from moderation, moderation analyses are often used to provide evidence about underlying causal mechanisms. The intuition is straightforward. For example, if a mechanism explains the causal relationship, we would expect individuals with different values of a given covariate to exhibit different treatment effects, as the covariate is assumed to moderate the effect. A major advantage of this approach is that it does not require direct measurement of the mediator. However, \citet{fu2023heterogeneous} find that valid inference of a mechanism from HTEs requires an exclusion assumption, ruling out the possibility of the covariate moderating through alternative pathways. This is intuitive because, otherwise, the observed HTEs cannot be attributed to the underlying mechanism. 

Moreover, in some cases, we cannot directly observe the outcome variable that is affected by the mechanism implied by the theory; instead, we can only measure a related outcome. For example, information about corruption involving the incumbent may influence voters’ relative evaluations of the incumbent and other candidates. While we cannot directly observe changes in utility, we can observe their behavioral manifestation in voting choices. In such cases, even when the exclusion assumption holds, the presence of HTEs does not necessarily indicate activation of the underlying mechanism. 

When the mediator is measured, another approach is to test whether the treatment affects the proposed mediator; \citet{blackwell2024assumption} clarify that a monotonicity assumption is needed to draw sharper conclusions. 


\section{Identification with Heterogeneous Effects}\label{sec:model}

In this section, we introduce the new identification strategy, which starts with synthetic causal decomposition under the counterfactual approach but emphasizes the mechanical process as the structural approach. Under the new decomposition, we then convert the difficult mediation problem into a simple linear regression problem. It turns out to be quite general and simple to identify the causal mediation under this new structure. We then compare our identification strategy with existing methods, emphasizing that it provides a useful alternative in settings where the sequential ignorability assumption is violated.

Recall that total causal effect can be decomposed into direct and indirect effects. As is standard, we begin with the ``no interaction effect'' situation in the main text ($\delta=\delta(1)=\delta(0)$ and $\eta=\eta(1)=\eta(0)$).\footnote{See SI \ref{si:int} for the discussion on the interaction effect.} The identification strategy starts with a straightforward transformation of this decomposition.
\begin{align}
    \tau &= \mathbb{E}[Y^i(1,M^i(0))- Y^i(0,M^i(0))]+ \mathbb{E} [Y^i(1,M^i(1)-Y^i(1,M^i(0))]\\
     &= \mathbb{E}[Y^i(1,M^i(1))- Y^i(0,M^i(1))]+ \mathbb{E} [Y^i(0,M^i(1)-Y^i(0,M^i(0))]\\
    &=\mathbb{E}[Y^i(1,M^i(1))- Y^i(0,M^i(1))] + \frac{\mathbb{E} [Y^i(0,M^i(1)-Y^i(0,M^i(0))]}{\mathbb{E}[M^i(1)-M^i(0)]} \times \mathbb{E}[M^i(1)-M^i(0)] \label{equ:p2}\\
    &:= \delta + \beta \gamma \label{equ:p3}
\end{align} The first two lines are two decompositions. In the line \eqref{equ:p2},we multiply and divide the average indirect effect $\eta=\mathbb{E} [Y^i(0,M^i(1)-Y^i(0,M^i(0))]$ by the same term $\gamma=\mathbb{E}[M^i(1)-M^i(0)]$. It is the average effect of treatment on the mediator of interests. We define $\frac{\eta}{\gamma}=\frac{\mathbb{E} [Y^i(0,M^i(1)-Y^i(0,M^i(0))]}{\mathbb{E}[M^i(1)-M^i(0)]}$ by $\beta$, which denotes the ratio of how pure indirect effect changes according to one unit change of $\gamma$. If researchers plan to use SEM to conduct mediation analysis, under the linear model \eqref{equ:sem1} and \eqref{equ:sem2}, we can easily see that $\beta=\tilde{\beta}$. Therefore, $\beta$ can be interpreted as the effect of the mediator $M$ on the outcome $Y$ \footnote{
Under linear SEM, we implicitly assume $\varepsilon_1$ and $\varepsilon_2$ are independent. From equation $\eqref{equ:sem1}$ and $\eqref{equ:sem2}$, we observe
\begin{equation}
\begin{aligned}
     \mathbb{E}[Y^i(0,M^i(1))]&=\alpha_1+\tilde{\beta}(\alpha_2+\gamma) \\
     \mathbb{E}[Y^i(0,M^i(0))]&=\alpha_1+\tilde{\beta}\alpha_2\\
\end{aligned}
\end{equation} Therefore, $\eta=\mathbb{E}[Y^i(0,M^i(1))]-\mathbb{E}[Y^i(0,M^i(0))]=\tilde{\beta} \gamma$, where $\mathbb{E}[M^i(1)-M^i(0)]=(\alpha_2+\gamma)-\alpha_2=\gamma.$}. Finally, we use simple notation to represent the final decomposition $\tau=\delta+\beta \gamma$.


If we can identify $\beta$ and $\gamma$, equivalently we can identify $\eta=\beta\gamma$. In most empirical studies, $\gamma$ and $\tau$ are easy to identify if treatment is as if random through careful research designs. The remaining part is to identify the parameter $\beta$. 

A critical insight in causal inference is the recognition that causal effects vary across populations and even among individuals. This implies that the $\gamma$ is a random variable (We will discuss how to get this sample in the section \ref{sec:des}). Therefore, equation \eqref{equ:p3} can be written as $\tau_k=\delta_k+\beta\gamma_k$. We use subscript $k$ to emphasize that they are random rather than fixed values. Next, we complete the transformation by adding and subtracting the expectation of $\delta_k$.
\begin{align}
    \tau_k &=\delta_k+\beta\gamma_k \\
    &= \mathbb{E}[\delta_k] + \beta\gamma_k + (\delta_k-\mathbb{E}[\delta_k]) \label{equ:lm2}\\
 \Rightarrow  \tau_k &= \mathbb{E}[\delta_k] + \beta\gamma_k + \varepsilon_k \label{equ:lm3}
\end{align} In the Line \eqref{equ:lm2}, we add and subtract the expectation of $\delta_k$; in line \eqref{equ:lm3}, we define $\varepsilon_k=(\delta_k-\mathbb{E}[\delta_k])$. In the structural equation, $\beta$, the ratio of the indirect effect to the treatment effect on the mediator is assumed to be constant. In a more general case, $\beta$ can also be random, and is denoted by $\beta_k$. Here, we focus directly on the effects. For strict frequentists, one can view this framework as analogous to a random-effects model in meta-analysis, which reconciles the “random” nature of $\gamma_k$ and $\tau_k$. A direct analogue is provided in Section~\ref{sec:des}.

    
    
    

Equation \eqref{equ:lm3} should be familiar to readers: it is a simple linear regression model. The key assumption for identifying $\beta$ is that the direct effect $\delta$ is uncorrelated with the effect of treatment on mediator $\gamma$. 

\begin{assumption}[Isolated Mechanism]\label{ass:cov}
    $Cov(\gamma_k,\delta_k)=0$.
\end{assumption}

The assumption requires that, $\gamma$, the average treatment effect on the mediator of interest, is not correlated with the average direct effect. If there exist multiple mechanisms, generally, we should interpret $\delta_k$ as effects from all other possible mechanisms, see SI \ref{app:mul} and \ref{app:ext}. Therefore, this assumption implies that the mechanism of interest is somewhat isolated from other mechanisms. A similar assumption has been proposed in the multiple instrumental variables literature \citep{kolesa2013estimation, bowden2015mendelian}. A related condition is also required for implicit mediation analysis using heterogeneous treatment effects \citep{fu2023heterogeneous}. Intuitively, if the covariate moderates multiple mechanisms, the observed HTEs may not provide information specifically about the mechanism of interest. Similarly, we require that the HTE in $\gamma$ reveals information solely about $\gamma$, rather than about other mechanisms. Technically, this assumption is equal to that of the traditional simple linear regression assumption $\mathbb{E}[\gamma_k\epsilon_k]=0$ and implies $Cov(\gamma_k,\varepsilon_k)=0$.\footnote{To see this, $\mathbb{E}[\gamma_k\epsilon_k]=\mathbb{E}[\gamma_k(\delta_k-\mathbb{E}\delta_k)]=Cov(\gamma_k,\delta_k)=0$ and $Cov(\gamma_k,\varepsilon_k)=Cov(\gamma_k,\delta_k-\mathbb{E}\delta_k)=Cov(\gamma_k,\delta_k)=0$. }

The validity of this assumption depends on the theoretical framework in question. To clarify, let us revisit our ongoing example. As we introduced earlier, there are four mechanisms—cost, civic duty, utility difference, and pivotality—that influence the decision to vote. If a researcher aims to examine the impact of mailings that encourage voters with the message "DO YOUR CIVIC DUTY—VOTE!" it is plausible to assume that such an intervention mainly affects the civic duty utility. Moreover, according to the theory, since $p_i b+D - C > 0$, other mediators are distinct and capture different mechanisms. Therefore, it is reasonable to assume that the effect of the canvassing treatment on civic duty is uncorrelated with its effects on other mechanisms. However, in general, assumption \ref{ass:cov} requires the belief that no other variable concurrently moderates both the mechanism of interest and other mechanisms. We can, however, allow for correlations among other mechanisms, such as the correlation between utility difference and pivotality, as illustrated in \ref{fig:voting}. In the section \ref{sec:ext}, we will provide more discussions and techniques on meeting the identification assumption in application.

In Proposition \ref{prop:consist}, we propose a simple estimator $\hat{\beta}$ to estimate the unknown $\beta$. 

\begin{proposition}\label{prop:consist}
Let $(\tau,\delta,\gamma)$ be random variables and as defined in \eqref{equ:p2} and \eqref{equ:p3}. Suppose we have the random sample $(\tau_k,\gamma_k)_{k \in K}$. Suppose $Var(\gamma_k)>0$ and Assumption \ref{ass:cov} holds.

Consider estimator $\hat{\beta}=\frac{\sum_{k=1}^K(\gamma_k-\overline{\gamma}_k)\tau_k}{\sum_{k=1}^K(\gamma_k-\overline{\gamma}_k)^2}$.

(1) If $\beta$ is a constant, then $\hat{\beta}$ $\xrightarrow{p} \beta$ as $K \rightarrow \infty$;

(2) If $\beta_k$ is a random variable, then $\hat{\beta} \xrightarrow{p} \mathbb{E}[\beta_k]$ as $K \rightarrow \infty$ under assumption $\beta_k\perp \gamma_k$

and thus $\eta$ is consistently estimated.
\end{proposition}

\begin{proof}
    All proofs are provided in the SI. 
\end{proof}

The estimator $\hat{\beta}$ is exactly the simple OLS estimator for the slope. The assumption $Var(\gamma_k)>0$ is technical; it guarantees that we have ``random'' observations of $\gamma$. The proposition indicates that if we assume the treatment effect on the mediator of interest is not correlated with the effects of other mechanisms, then $\beta$ can be consistently estimated. Combined with the information of $\gamma$, the overall average mediation effect $\mathbb{E}[\eta_k]$ can be estimated by $\hat{\beta} \sum_{k=1}^K \gamma_k \mathbb{P}(\gamma_k)$, where $\mathbb{P}(\gamma_k)$ can be consistently estimated by the proportion of sample size $k$ relative to the total sample size. The \href{https://github.com/Jiawei-Fu/mechte}{R package} will return all essential statistics. It is important to note that the estimates produced by other mediation analysis techniques, which do not account for heterogeneity, can be interpreted as average effects over implicit heterogeneous effects. 

\begin{remark}
    Whether $\beta$ is a constant can actually be tested, and we will demonstrate this in the application.
\end{remark}

\begin{remark}[$\beta_k \perp \gamma_k$]
    In the proposition \ref{prop:consist}, if $\beta_k$ varies with $k$, we need to assume $\beta_k \perp \gamma_k$ to ensure $\hat{\beta}$ is a consistent estimator of $\beta$. It may be a strong assumption in some cases. A simple solution is to use nonparametric regression to approximate the correlation between $\beta$ and $\gamma$. According to the Stone–Weierstrass theorem, every continuous function on a compact set can be uniformly approximated by a polynomial function. For example, consider a $p^{th}$ order polynomial, $\beta_k=\theta_0+\theta_1\gamma_k+\theta_2\gamma^2_k+...+\theta_p\gamma^p_k$. Substituting it in the model \eqref{equ:lm3} yields $\tau_k=\mathbb{E}[\delta_k]+(\theta_0+ \theta_1\gamma_k+...+\theta_p\gamma_k^p)\gamma_k+\varepsilon_k$. The parameter of interest is now characterized by $\theta_p$ rather than $\beta$.
\end{remark}

Readers may be concerned that we are assuming a linear relationship between $\tau_k$ and $\gamma_k$. Actually, we do not assume such linearity; $\tau_k=\mathbb{E}[\delta_k] + \beta\gamma_k + \varepsilon_k$ is the structural model derived from causal decomposition. There is an important distinction between this structural model and statistical linear regression models. First, in most cases, people assume the linear statistical relationship between data, that is, $\tau_k$ and $\gamma_k$ here. Nevertheless, our model $\tau_k=\mathbb{E}[\delta_k] + \beta\gamma_k + \varepsilon_k$ is naturally guaranteed by the nature of the causal effect. In the counterfactual framework, we can always additively decompose the total causal effect into two pieces. Second, in statistical applications, people assume the expectation of the error term in their population model is zero: $\mathbb{E}[\varepsilon_k]=0$. However, here, this property is guaranteed by construction, not by assumption: $\mathbb{E}[\varepsilon_k]=\mathbb{E}[\delta_k]-\mathbb{E}[\delta_k]=0$. Because of this property, in contrast to OLS, the unbiasedness of our estimator requires a slightly weaker mean independence assumption ($\mathbb{E}[\delta_k|\gamma_k]=\mathbb{E}[\delta_k]$), rather than the zero conditional mean assumption ($\mathbb{E}[\delta_k|\gamma_k]=0$). We summarize this result in SI \ref{app:unbias}.

What are the main advantages of our identification strategy? First, it does not require that mediator is ignorable. In other words, we allow unobserved confounders that simultaneously affect the mediator and the outcome variable (i.e., $U_2$ in the Figure \ref{fig:dag2}). As mentioned in the section \ref{sec:idenassump}, current methods cannot efficiently address this unconfoundedness problem without further assumptions. 


Second, we allow researchers to simultaneously estimate both treatment and mediation effects. The causal mediation, is simply a byproduct after identifying the treatment effects ($\tau$ and $\gamma$). We do not need other advanced techniques to identify the indirect effect except simple OLS. We will introduce exact estimation methods and research designs in the next section. We can use both aggregate-level data and individual-level data to get the causal mediation effect. Therefore, we believe our methods can be applied in a variety of empirical studies when sequential ignobility does not hold.

\section{How to Obtain Heterogeneous Effects?}\label{sec:des}

To implement the strategy, we need a random sample of treatment effects, $(\gamma_K,\tau_K)$. It is common for the same treatment to generate varying (average) treatment effects across different populations, such as those in different countries or areas. This is similar to getting multiple effects when conducting meta-analysis. Rigorously, we assume there are $k$ independent trials that generate $k$ study-level effects $(\gamma_k,\tau_k)$, where $\gamma_k$ is assumed to come from a super-population (normal) distribution. In the random-effects model, $\gamma_k= \mu + \Delta_k +e_k$, where $\mu$ is the overall effect, $\Delta_k$ is the deviation of $k's$-effect from the overall effect, and $e_k$ is the sampling error. As is standard in meta-analysis, we require that $\Delta_k$ be independent of study $k$, or exchangeable in the language of Bayesian Statistics \citep{higgins2009re}. This assumption is critical for statistical inference; therefore, we will proceed under the premise that it holds in the subsequent analysis.

Under the above model, practically, how can we find study-level effects $(\gamma_k,\tau_k)$? If the same research is conducted several times in different locations, we naturally obtain data from those studies. If not, it is still possible to `create' multiple studies within a single study. In this section, we introduce two general research designs.

\textbf{Multiple Treatment Meta Design.} The key idea is that we can adjust the treatment to induce heterogeneous effects. For example, in the Get-Out-The-Vote (GOTV) experiments, researchers could modify the treatment by having some groups receive one mail, while other groups receive two or three mails. See Figure \ref{fig:multiple_treat}. Given the same treatment, varying the intensity or magnitude of the treatment can likely induce different effects. Formally, we still use $G_k \in \{T_1,T_2,...,T_l\}$ to denote different sub-types of the treatment. If individual $i$ belongs $G_k=T_j$, it means the individual receives treatment intensity $T_j$.

\begin{figure}[!h]
    \centering
    \includegraphics[width=.6\textwidth]{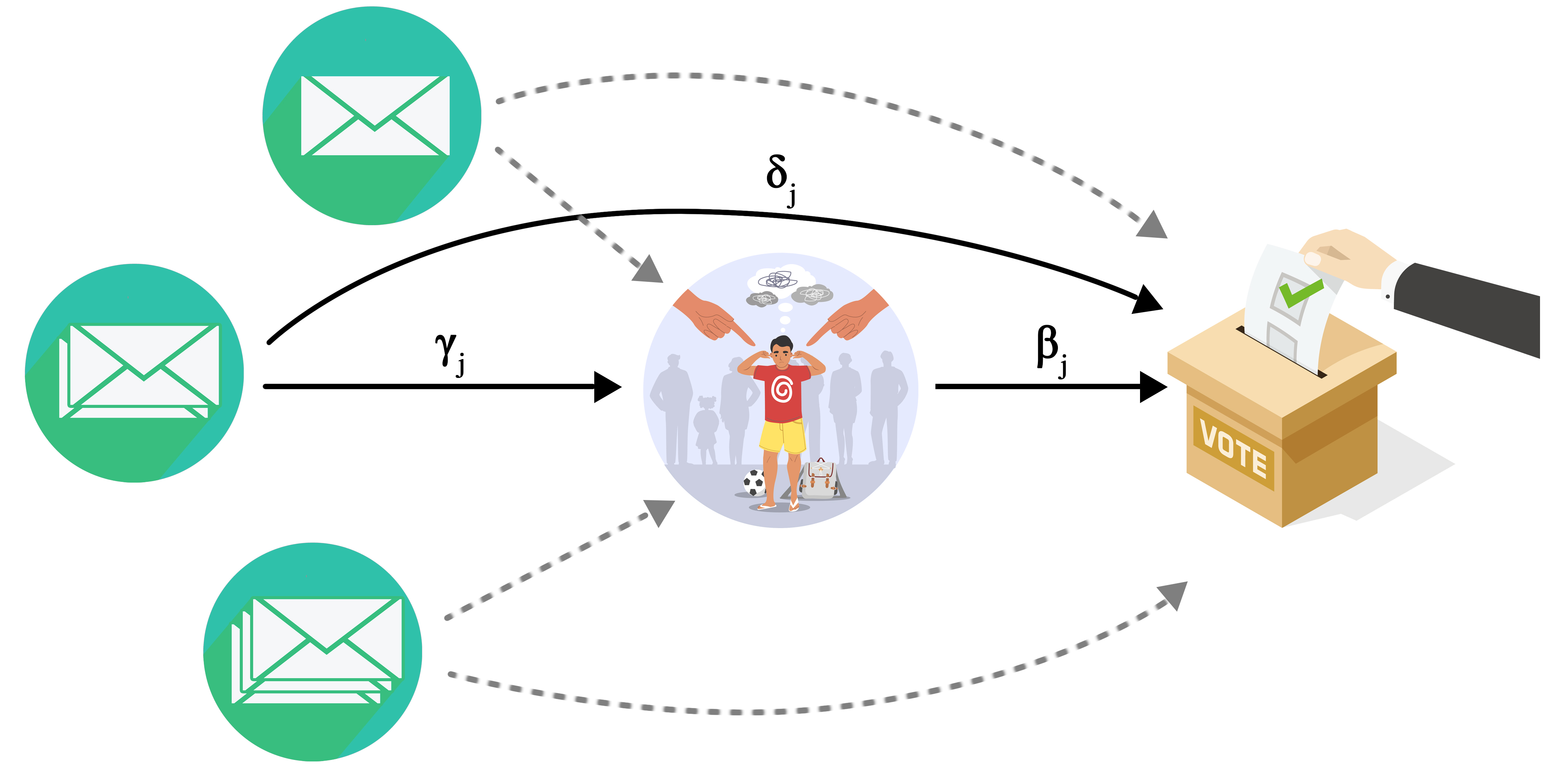}
    \caption{Multiple Treatment Meta Design. In experiment, we can adjust the treatment to induce heterogeneous effects.}
    \label{fig:multiple_treat}
\end{figure}

\textbf{Heterogeneous Subgroup Design.} This design exploits the heterogeneous effects from different population. Naturally, the same treatment may generate different (average) treatment effects for different population, for example, population in different country or area. To be concrete, for example, suppose researchers desire to understand how the mailing affects turnout through social pressure in the GOTV experiment \citep{gerber2008social}, as shown in Figure \ref{fig:hetero_popu}. Formally, each individual is characterized by a vector of pre-treatment covariates $X=(X_1,X_2,...,X_l)$ that can moderate treatment effects on the outcome and the mediator. We can subsequently define several subgroups $G_k$, where $k \in\{1,2,...,K\}$ according to $X$. Suppose $X_1$ is gender, and $X_2$ is age. We can define group $G_1=\{X_1=Male, X_2>30\}$, comprising individuals who are male and older than 30. Each individual $i$ should belong to only one group. An assumption is that for each group, treatment generates different and independent average treatment effects $\tau$ and $\gamma$. How should we identify these groups? If several similar studies are conducted in geographically different areas, similar to meta-analysis, then those  studies automatically generate a sample of effects. Alternatively, a data-driven method such as causal tree or forest can be used \citep{wager2018estimation} \footnote{Technically, the data obtained from such a design represent conditional average effects. The analysis remains valid under the assumption on $\beta$, as we take the average across subgroups.}.

\begin{figure}[!h]
    \centering
    \includegraphics[width=.6\textwidth]{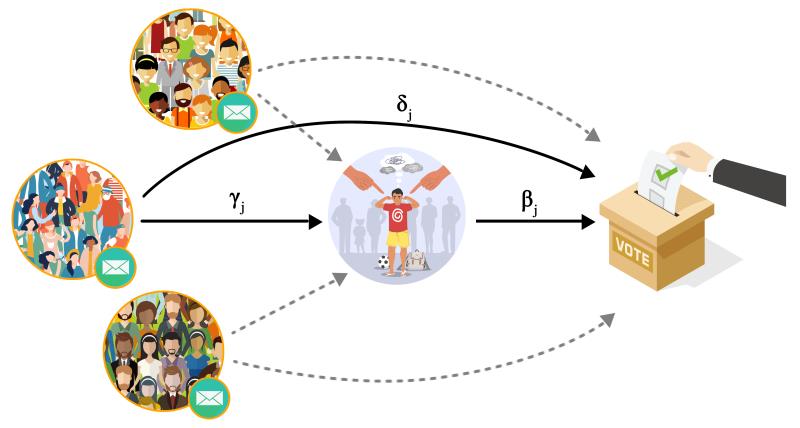}
    \caption{Heterogeneous Subgroup Design. Heterogeneous effects
can be derived from different population.}
    \label{fig:hetero_popu}
\end{figure}


\textbf{Synthetic Methods.} People may incorporate these two designs and define finer subgroups. The incorporated subgroup $G_k = \{T\in \textbf{T},X_1 \in \textbf{X}_1,X_2 \in \textbf{X}_2,...,X_l\in \textbf{X}_l\}$ is defined by treatment types and covariates, where $\textbf{X}_l$ denotes a set possible values of $X_l$. For example, in the GOTV design, for each intensity of treatment, we can find subgroups defined by covariates. A possible subgroup could be $G_k=\{Single Mail, X_1=Male, X_2>30\}$. If individual $i$ is in this group, it implies that individual $i$ is male, older than 30, and receive treatment phone call.

Regardless of the method used, the objective is the same: to obtain data suitable for regression analysis. Our goal is not to identify the true, unobserved subgroups. Rather, we seek to generate variation in the estimated effects $\hat{\gamma}_k$, which—by the causal decomposition—correspond to associated estimates $\hat{\tau}_k$. If observations are partitioned into subgroups at random, the resulting $\hat{\gamma}_k$ are likely to be similar across $k$. Although such estimates can still be used, the limited variation in $\hat{\gamma}_k$ leads to large standard errors. We therefore prefer subgroup constructions that induce substantial heterogeneity across $k$. Importantly, these subgroups need not correspond to the true underlying subpopulations. Identifying the correct latent subgroups is an interesting and important research problem in its own right, but it is not required for our approach.

\section{Estimation and Inference}\label{sec:infer}

Once we obtain the heterogeneous effects $(\hat{\gamma}_k,\hat{\tau}_k)$, we are ready to estimate $\beta$ and, consequently, the causal mediation effects. Because the second-stage regression uses estimated group-level effects,
\((\widehat\tau_k,\widehat\gamma_k)\), rather than the infeasible quantities
\((\tau_k,\gamma_k)\), the estimator is a generated-regressor estimator. To ensure that the first stage has no first-order effect on the second stage, we require the first-stage heterogeneous effects $(\hat{\gamma}_k,\hat{\tau}_k)$ to be regular \(n_k^{1/2}\)-consistent treatment-effect estimators. Formally, for each group $k$, let $\widehat\gamma_k = \gamma_k + u_{k}$, and $\widehat\tau_k = \tau_k + v_{k}$, where $u_k$ and $v_k$ denote the estimation errors.

\begin{assumption}\label{ass:first_stage_rate}
    $\sum_{k=1}^K(u^2_k+v^2_k)=O_p(\sum_{k=1}^K\frac{1}{n_k})$ . 
\end{assumption}

If the groups are identified ex ante, then this rate condition is generally satisfied under standard regularity conditions. When subgroups are selected using a data-driven method, honesty or sample splitting provides a convenient sufficient condition. Specifically, we partition the data into two disjoint subsamples, $I_1 \cup I_2$. The first subsample $I_1$ is used to identify the $K$ subgroups, while the second subsample $I_2$ is used to estimate $(\hat{\gamma}_k,\hat{\tau}_k,\hat{\sigma}^2_{uk})$, where $\hat{\sigma}^2_{uk}$ denotes the estimated asymptotic variance of $\hat{\gamma}_k$. In practice, $k$-fold cross-fitting is recommended, as is standard in the current DML literature. This separation ensures that,
conditional on the selected partition, the group-level estimators behave like standard treatment-effect estimators.

\begin{proposition}\label{prop:twostep_ols}
Suppose $(\tau_k,\gamma_k)$ satisfies the decomposition \eqref{equ:lm3}: $\tau_k = \mathbb{E}[\delta_k]+\beta\gamma_k+\varepsilon_k$.
Define the plug-in second-stage OLS estimator
\begin{equation}
    \widehat\beta_K
    =
    \frac{
        \sum_{k=1}^K
        (\widehat\gamma_k-\bar{\widehat\gamma})
        (\widehat\tau_k-\bar{\widehat\tau})
    }{
        \sum_{k=1}^K
        (\widehat\gamma_k-\bar{\widehat\gamma})^2
    },
    \label{eq:plugin_ols}
\end{equation}
where $
    \bar{\widehat\gamma}
    =
    \frac{1}{K}\sum_{k=1}^K \widehat\gamma_k$ and $\bar{\widehat\tau}
    =
    \frac{1}{K}\sum_{k=1}^K \widehat\tau_k$. Suppose $Q := \mathbb E[(\gamma_k-\mathbb{E}[\gamma_k])^2] > 0$, and that
$\Omega := \mathbb E[(\gamma_k-\mathbb{E}[\gamma_k])^2\varepsilon_k^2] < \infty$. Assume that Assumption \ref{ass:cov} and \ref{ass:first_stage_rate} hold. Define $V_\beta=\frac{\Omega}{Q^2}$. If the rate condition $\sum_{k=1}^K \frac{1}{n_k} \to 0$ holds, then
    \begin{equation}
        \sqrt K(\widehat\beta_K-\beta)
        \xrightarrow{d}
        \mathcal N\left(0, V_\beta\right)
        \label{eq:asymptotic_normality}
    \end{equation}
\end{proposition}

The rate condition $\sum_{k=1}^K \frac{1}{n_k} \to 0$ ensures that this total first-stage noise vanishes asymptotically. In the balanced case, the condition $ \sum_{k=1}^K \frac{1}{n_k} \to 0$ reduces to \(K/n_k\to 0\). In general, this condition requires that the sample size within each group grow sufficiently fast so that the accumulated first-stage estimation error disappears asymptotically. A consistent heteroskedasticity-robust estimator of \(V_\beta\) is
\begin{equation}
    \widehat V_\beta
    =
    \frac{
        K^{-1}\sum_{k=1}^K
        (\widehat\gamma_k-\bar{\widehat\gamma})^2
        \widehat\varepsilon_k^2
    }{
        \left[
            K^{-1}\sum_{k=1}^K
            (\widehat\gamma_k-\bar{\widehat\gamma})^2
        \right]^2
    },
    \label{eq:robust_vbeta}
\end{equation}
where $
    \widehat\varepsilon_k
    =
    \widehat\tau_k-\widehat\alpha_K-\widehat\beta_K\widehat\gamma_k$, and $ \widehat\alpha_K
    =
    \bar{\widehat\tau}-\widehat\beta_K\bar{\widehat\gamma}$.
    
    Proposition \ref{prop:twostep_ols} is stated under this asymptotic framework, in which subgroup sizes grow such that the first-stage estimation error in \(\widehat\gamma_k\) and \(\widehat\tau_k\) vanishes. In finite samples, however, when group sizes are small, one may still be concerned about measurement-error attenuation arising from first-stage estimation noise. In such settings, we suggest using the well-known simulation--extrapolation (SIMEX) method \citep{cook1994simulation}. In SI \ref{si:twostep}, we generalize and adapt standard asymptotic normality results to our setting, which involves two-step estimation with non–i.i.d.\ data.

Using either approach, we obtain a well-behaved estimator for $\beta$. Our primary estimand, however, is the overall average indirect effect, defined as the weighted average of $\beta_k \gamma_k$ across the $k$ subgroups. The variance of the estimator can be derived through Delta method. As shown in SI \ref{si:twostep}, under further sample splitting, the asymptotic variance equals $(\frac{1}{K}\sum_{k=1}^K \hat{\gamma}_k)^2V_{\beta}$ under a further sample splitting. To achieve better finite-sample performance, we instead employ an intersection–union test. For simplicity, let $\lambda_k$ denote the population proportion associated with $\gamma_k$, and let $\hat{\lambda}_k$ represent its consistent estimator.
The following proposition summarizes how to test the null effect of the overall average indirect effect based on $p$-values. Let $\hat{\sigma}_{\beta}$ be the standard error. $\Phi_z$ is the cumulative distribution function of the standard normal distribution.

\begin{proposition}\label{prop:inference}
Assuming the conditions in the proposition of \ref{prop:meta} hold, let $p_{\beta}=2*\Phi_Z(-|\frac{\hat{\beta}}{\hat{\sigma}_{\beta}}|)$ and $p_{\gamma}=2*\Phi_Z(-|\frac{\sum_{k=1}^K\hat{\gamma}_k \hat{\lambda}_k}{\sqrt{\sum_{k=1}^K \hat{\lambda}^2 \hat{\sigma}^2_{uk}}}|)$. 

Consider the null hypothesis that the overall average indirect effect is zero. 
Then, rejecting the null hypothesis if $p_{\beta}\leq \alpha$ and $p_{\gamma}\leq \alpha$ constitutes an asymptotic level $\alpha$ test.
\end{proposition}

In fact, $p_{\beta}$ and $p_{\gamma}$ are $p$-values for two sub-parts in the intersection-union test. From the above proposition \ref{prop:inference}, it follows that the overall $p$-value associated with the null effect of the overall indirect effect is defined as $\max [p_{\beta}, p_{\gamma}]$. The next proposition provides an conservative confidence interval based on the intersection-union test. We use $\hat{\gamma}_0$ to denote $\sum_{k=1}^K\hat{\gamma}_k \hat{\lambda}_k$.

\begin{proposition}\label{prop:inference2}
    Assuming the same conditions in the proposition of \ref{prop:inference} hold. Define   
    $$
    \begin{aligned}
        a_1 &= \hat{\gamma}_0 - \Phi^{-1}_Z (\frac{1+\sqrt{1-\alpha}}{2})* \sqrt{\sum_{k=1}^K \hat{\lambda}^2 }\hat{\sigma}^2_{uk}, \;\;\;\
        a_2=\hat{\gamma}_0 + \Phi^{-1}_Z (\frac{1+\sqrt{1-\alpha}}{2})* \sqrt{\sum_{k=1}^K \hat{\lambda}^2 \hat{\sigma}^2_{uk}}\\
        a_3&=\hat{\beta} - \Phi^{-1}_Z (\frac{1+\sqrt{1-\alpha}}{2})* \hat{\sigma}_{\beta},\;\;\;\;
        a_4=\hat{\beta} + \Phi^{-1}_Z (\frac{1+\sqrt{1-\alpha}}{2})* \hat{\sigma}_{\beta},
    \end{aligned}
    $$ $\overline{a}=\max[a_1a_3,a_1a_4,a_2a_3,a_2a_4]$, and $\underline{a}=\min[a_1a_3,a_1a_4,a_2a_3,a_2a_4]$.

The at least $(1-\alpha)\%$ confidence interval for overall average indirect effect is given by $[\underline{a},\overline{a}]$, ensuring $\mathbb{P}[\underline{a} \leq \mathbb{E}[\eta_k] \leq \overline{a}] \ge 1-\alpha$, asymptotically.
 
\end{proposition}


Based on both propositions, conducting sub-group inference is also straightforward. We encapsulate estimation and inference in the \href{https://github.com/Jiawei-Fu/mechte}{R package}.
In the package, we also include Cochran's Q and Higgins \& Thompson’s $I^2$ from meta-analysis to assess whether $\gamma_k$ values exhibit true heterogeneity rather than random error.

\section{Simulation and Application}\label{sec:appli}

In this section, we employ Monte Carlo simulations to evaluate the effectiveness of our methods by comparing them with the current methods under sequential ignorability. Furthermore, we apply our methodology to real data from two distinct experiments—one using aggregate data and the other using individual-level data—to illustrate its application in real studies.

\subsection{Simulation}\label{sec:sim}

We generate heterogeneous treatment effects for each individual $i$ by assuming $10$ subgroups using the following simple model:

$$
\begin{aligned}
M_{ki} &= 1 + \gamma_k*T_i+ \omega *u_i+ \epsilon^M_i\\
Y_{ki} &=1+T_i+\beta M_{ki}+ \omega*u_i + \epsilon^Y_i
\end{aligned}
$$ Treatment $T_i$ is randomly drawn from a standard normal distribution, along with two error terms $\epsilon$. The average treatment effect on the mediator $M_i$ varies across groups, with $\gamma_k \in \{1,2,3,4,5,6,7,8,9,10\}$. The effect of the mediator on the outcome is $\beta$. We also consider an `unobserved' confounder $u_i$ that simultaneously affects the mediator and outcome, with magnitude $\omega \in \{0,0.5,1,1.5,2\}$. 

We first set $\beta=0$ so that the AMCE is zero. In Table \ref{tab:sim2}, we report the empirical distribution of p-values and the $95\%$ confidence intervals (averaged over simulations) for both methods across different values of $\omega$. Confidence intervals are constructed using the method introduced in Section \ref{sec:infer}. It is evident that when $\omega=0$, the sequential ignorability assumption holds. In the first row, therefore, both methods yield estimates that are very close to the theoretical AMCE of 0. As $\omega$ increases, the sequential ignorability assumption is violated. Our method remains robust to this violation, in the sense that the confidence interval continues to cover 0 and the Type I error rate in the first column remains close to the nominal level. However, for methods relying on sequential ignorability, the rejection rate increases and the confidence interval no longer covers the true value of zero.

\begin{table}[!ht]
\centering
\begin{tabular}{c|ccc|cccc}
\hline 
\hline 
   & \multicolumn{3}{c|}{HTE Method } &\multicolumn{3}{c}{Traditional Method}\\
  $\omega$ & \multicolumn{3}{c|}{with SIMEX} &\multicolumn{3}{c}{ with Sequential Ignorability}\\
    \hline
 & P(p$<$0.05) & CI low & CI up & P(p$<$0.05) & CI low & CI up  \\ 
0.00 & 0.05 & -0.18 & 0.19 & 0.03 & -0.15 & 0.16 \\ 
  0.50 & 0.07 & -0.20 & 0.23 & 0.35 & -0.02 & 0.32 \\ 
  1.00 & 0.06 & -0.26 & 0.31 & 1.00 & 0.34 & 0.75 \\ 
  1.50 & 0.05 & -0.32 & 0.41 & 1.00 & 0.84 & 1.32 \\ 
  2.00 & 0.06 & -0.39 & 0.53 & 1.00 & 1.39 & 1.93 \\ 
   \hline
   \hline 
\end{tabular}
\caption{The table reports averages across simulations. Column $P(p<0.05)$ reports the empirical proportion of p-values below 0.05. When $\omega = 0$, both the HTE method and the sequential ignorability method yield valid estimates. When $\omega \neq 0$, the sequential ignorability assumption is violated. In this case, we observe that the confidence interval of the HTE method still covers the true value of 0, and the Type I error rate is close to the nominal level.}
    \label{tab:sim2}
\end{table}

Next, we study the power of our method. To do so, we increase $\beta$ so that the AMCE changes from 10\% to 30\% of the ATE. We do not further increase this percentage because the power already reaches 1. As shown in Figure \ref{fig:mech_power}, when the AMCE exceeds 20\% of the ATE, the power of our method reaches approximately 80\%.This suggests that our method has strong power to detect relatively small AMCEs in practice.

\begin{figure}[!h]
     \centering
 \includegraphics[width=.6\textwidth]{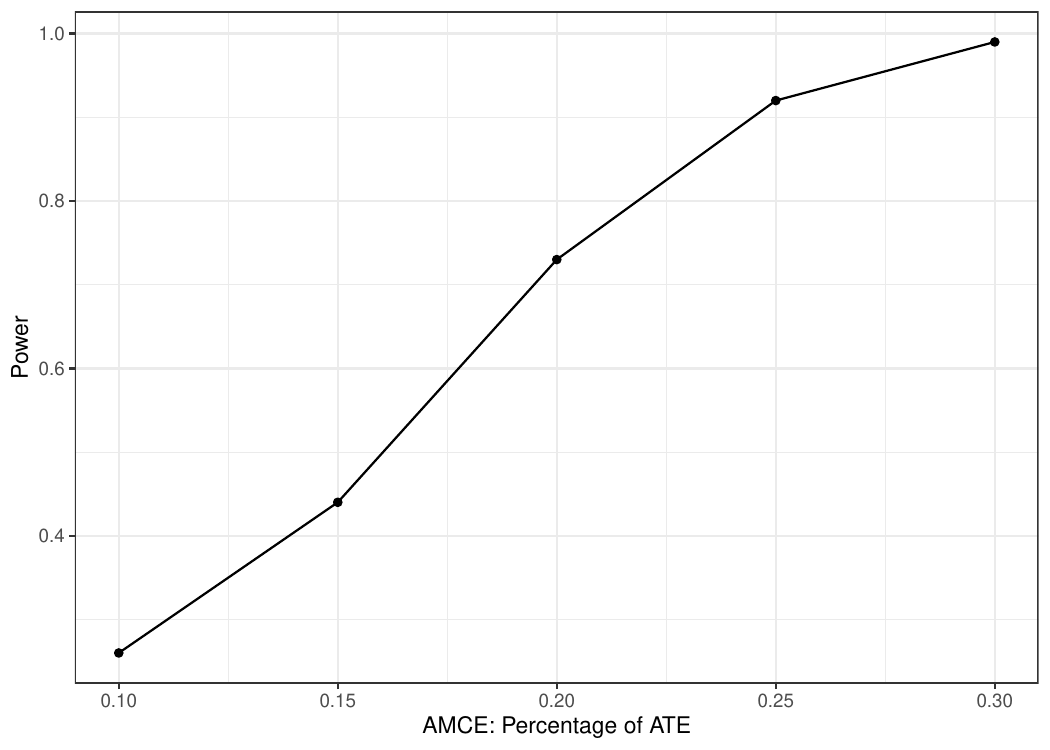}
     \caption{The figure shows the statistical power as the AMCE increases from 10\% to 30\% of the average total treatment effect. The power quickly reaches 1 when the AMCE is around 30\% of the total effect.}
     \label{fig:mech_power}
 \end{figure}

We also explore how the number of groups and the group size influence statistical power using a similar DGP as before. The results are shown in Figure \ref{fig:power}. In general, both increasing the number of groups and increasing the group size enhance statistical power. However, in this simulation, increasing the number of groups has a more pronounced effect than increasing the group size. This suggests that heterogeneity provides more information than merely increasing the precision of the estimates. In practice, researchers are advised to conduct similar power analyses to optimize their research designs.



\subsection{Application I: Governance on Resources (Aggregate level Data)}

Evidence in Governance and Politics (EGAP) \footnote{https://egap.org/} funds and coordinates multiple field experiments on different topics across countries. This collaborative research model is called ``Metaketa Initiative.'' In Metaketa III, they examine the effect of community monitoring on common pool resources (CPR) governance. To causally answer this question, \citet{slough2021adoption} conducted six harmonized experiments with the same `meta' treatment (community monitoring) but heterogeneous CPRs and treatment sub-types, as shown in the table \ref{tab:exp_sum}.

\begin{table}[!h]
    \centering
    \resizebox{\textwidth}{!}{
    \begin{tabular}{c|cccccc}
         & Brazil &China &Costa Rica&Liberia&Peru&Uganda \\ \hline
    Resource  & Groundwater & Surface water & Groundwater & Forest & Forest & Forest \\
         
      \textbf{Components of treatment}   & &&&&&\\
      Community workshops & \checkmark & - &\checkmark &\checkmark & \checkmark & \checkmark \\
      Monitor selection, training,incentives & \checkmark & \checkmark &\checkmark &\checkmark & \checkmark & \checkmark\\
      Monitoring of the resource &\checkmark & \checkmark &\checkmark &\checkmark & \checkmark & \checkmark\\
      Dissemination to citizens &\checkmark & \checkmark &\checkmark &\checkmark & \checkmark & \checkmark\\
      Dissemination to management bodies &- & (alternative arm) &\checkmark &\checkmark* & \checkmark* & \checkmark*\\ \hline
    \end{tabular}
    }
    \caption{Summary Table from \citet{slough2021adoption}. * In the forest studies, the community constitutes at least one of the possibly overlapping management bodies.}\label{tab:exp_sum}
\end{table}

In their study, the authors report effects on multiple outcome variables, which include resource use, user satisfaction, user knowledge about community’s CPRs, and resource stewardship. They also investigate the underlying mechanism: how monitoring affects those outcomes through different channels. However, their analysis is limited to examining the treatment effects on mediators, which does not necessarily delineate the precise causal mediation effects. We intend to use their data to illustrate how to apply our causal mediation analysis with aggregate-level data; to be specific, we ask ``How does the treatment (i.e. monitoring) affect user knowledge about community CPRs through altered perceptions of sanction likelihood for CPR misuse?"

The six experiments naturally provide us with six subgroups. While it is feasible for researchers to further segment subgroups within each experimental site, our focus will be on these six primary subgroups. To estimate ACME, we require specific data: (1) the average treatment effects on both the outcome and the mediator, and (2) the standard errors associated with these effects. These data points are represented in the two-dimensional Figure \ref{fig:cmp_est}, where each dot represents an estimate and each line indicates the $95\%$ confidence interval (CI). Generally, most confidence intervals cross zero, particularly for the treatment effects on the mediator (sanction), as illustrated in Figure \ref{fig:ori}. Here, no single estimate significantly deviates from zero at the 0.05 level.

However, the data reveal a clear pattern: an increased treatment effect on the perception of sanctions correlates with an increased effect on knowledge about the resource. To quantify the mediation effect, we estimate $\beta$, with the estimated slope being $\beta = 0.68$ and the $p$-value =0.049. This estimate is also depicted by the green line in Figure \ref{fig:cmp_est}. The figure is also useful for determining whether $\beta$ is constant. If it is indeed constant, we should expect the data points to display a relatively linear relationship. As previously mentioned, in the structural approach to mediation analysis, we can generally interpret $\beta$ as the effect of mediator (sanction) on the outcome (knowledge). This confirms our expectation that an increased perception of sanctions enhances the incentive to gain more knowledge about common-pool resources.

\begin{figure}[!tbp]
  \centering
  \begin{minipage}[b]{0.4\textwidth}
    \includegraphics[width=\textwidth]{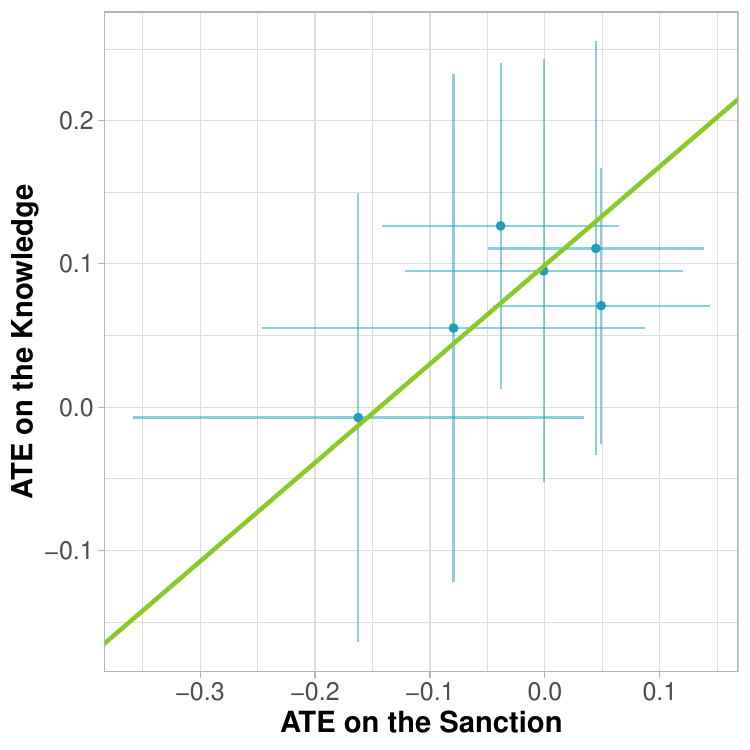} 
    \caption{Average treatment effects on the mediator and outcome with standard errors in Six experiments.}\label{fig:cmp_est}
  \end{minipage}
  \hfill
  \begin{minipage}[b]{0.4\textwidth}
    \includegraphics[width=\textwidth]{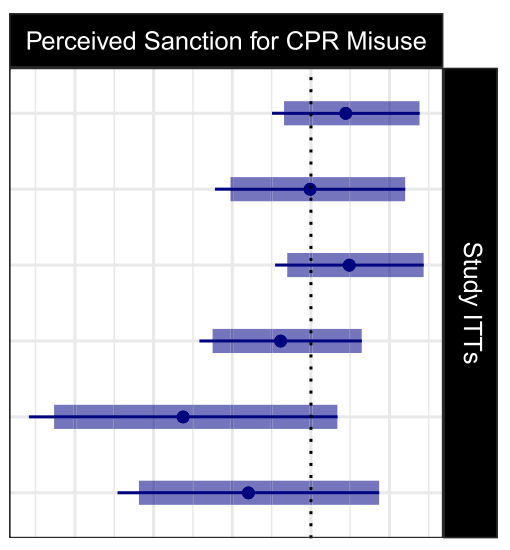}
    \caption{The original figure from \citet{slough2021adoption} on the average treatment effects on the mediator.} \label{fig:ori}
  \end{minipage}
\end{figure}


To obtain the average causal mediation effect, we need to multiply $\beta$ by $\gamma$ (the ATE on the mediator) and possibly weight this by the sample size in each experiment. We find that the estimated ACME is $0.01$. Since all $\gamma$ in six experiments are not statistically significant at $0.05$ level, it is also challenging to achieve a significant ACME. However, the positive ACME does suggest that the monitoring enhances knowledge about CPRs can be explained by a mediation effect through the perception for CPR misuse.

\subsection{Application II: Information Effect on Voting (Individual-level data)}

Accountability is a cornerstone of democracies and is fundamental to good governance. However, in reality, voters often lack sufficient information about politicians' performance. Many organization and civil society groups have dedicated efforts to disseminate such information to the electorate. A pivotal question arises: ``Do informational interventions influence voters' behaviors and thereby promote accountability? If yes, what is the key mechanism?" Numerous field and survey experiments have sought to quantify this treatment effect; yet the findings are inconclusive \citep[see][]{incerti2020corruption,dunning2019voter}.  
Furthermore, our understanding of how information influences voting behavior is limited. In a few experiments, researchers have measured intermediate outcomes and explore potential mechanisms. Nevertheless, these intermediate outcomes are not ignorable given the treatment status, making the estimation of mediation effects challenging.  In this section, we will illustrate the use of our method in one field experiment from Benin, demonstrating how to identify and estimate the causal mediation effect in an information experiment using individual-level data.

Around 2015 National Assembly elections in Benin, as a part of Metaketa I, \citet{adida2019under} randomly disseminated information about the performance of incumbent legislators to voters through videos. These videos provided official data on four key performance dimensions: (1) attendance rate at legislative sessions, (2) frequency of posing questions during these sessions, (3) committee attendance rate, and (4) productivity of committee work. One of the primary outcome variables was individual voting choice, which was captured via baseline and endline surveys. The surveys also gathered intermediate variables, such as voters' perceptions of the incumbents' effort/hardworking. Overall, the intervention did not significantly affect the incumbents' vote shares, aligning with the results from most other field experiments \citep{dunning2019voter}. However, a subsequent meta-analysis highlighted a notable correlation between voters' perceptions of effort and support for incumbents \citep[p354]{dunning2019information}. As emphasized by authors, this correlation does not illuminate any causal relationship due to the design. Nevertheless, it indicates a potential indirect effect of information on voting behavior mediated by perceptions of hard work. Thus, we intend to apply our method to their individual-level data to directly estimate this causal mediation effect.

To apply our method to individual-level data, the initial step involves identifying potential heterogeneous subgroups. This identification can be achieved through data-driven methods. We employ the widely-used causal tree approach to detect these subgroups, estimating the heterogeneous treatment effect on the mediator (effort) using individual pre-treatment covariates, such as age, gender, wealth, and political attitudes.\footnote{For details, please refer to the replication materials.} As depicted in Figure \ref{fig:subgroup}, the informational effect on the perception of effort varies according to factors like coethnic, education, age, wealth status etc. With different algorithms, the identified subgroups may vary. For example, we can obtain alternative subgroup classifications, as shown in the SI \ref{sec:app_ind}. However, if the identification assumptions hold, $\gamma_k$ and $\tau_k$ should capture the same underlying information, and thus the estimate of $\beta$ remains nearly identical. This finding is also empirically confirmed.

\begin{figure}
    \centering
    \includegraphics[width=0.9\textwidth]{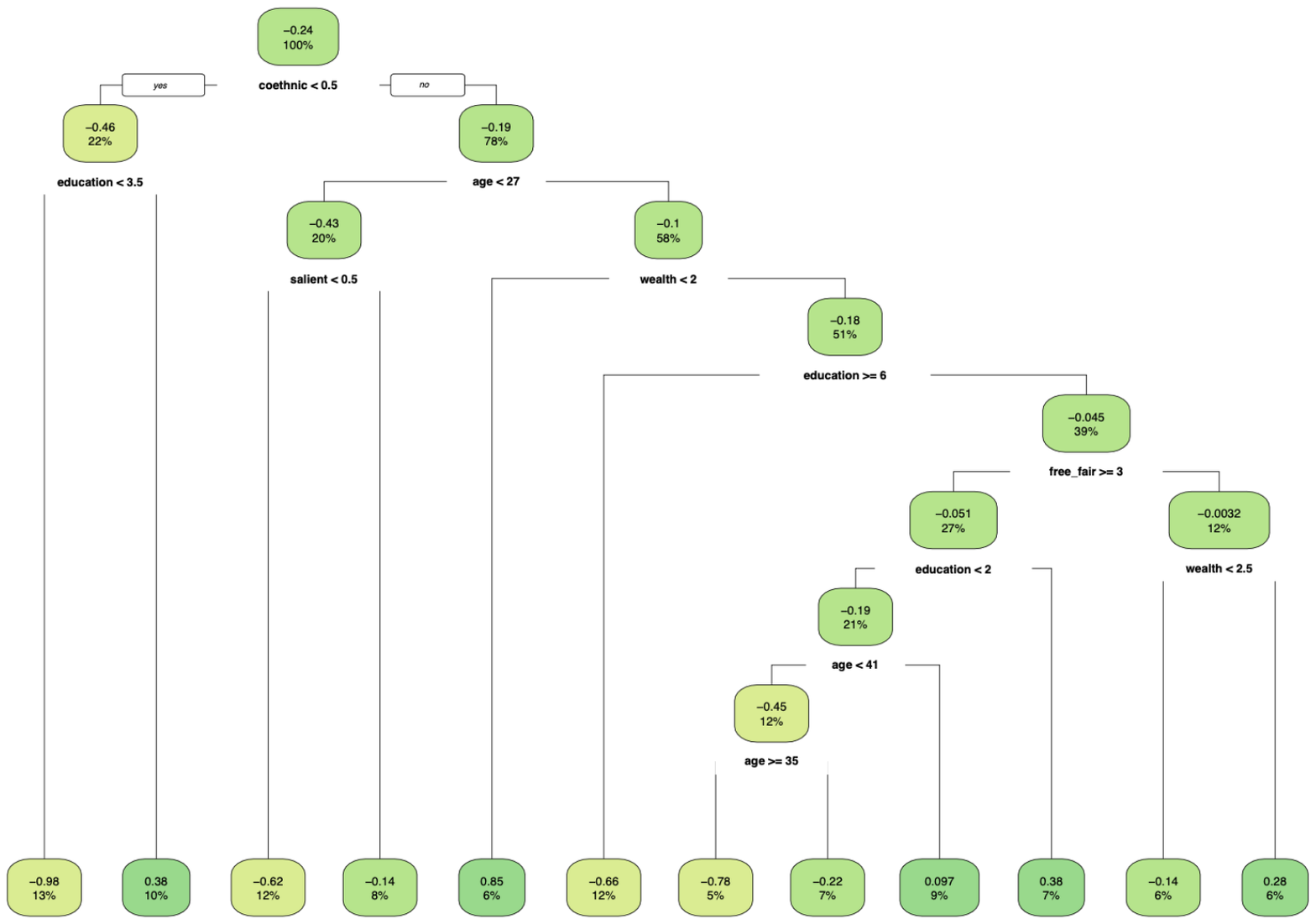}
    \caption{Subgroups detected by Causal trees}
    \label{fig:subgroup}
\end{figure}

\begin{figure}
    \centering
    \includegraphics[width=0.7\textwidth]{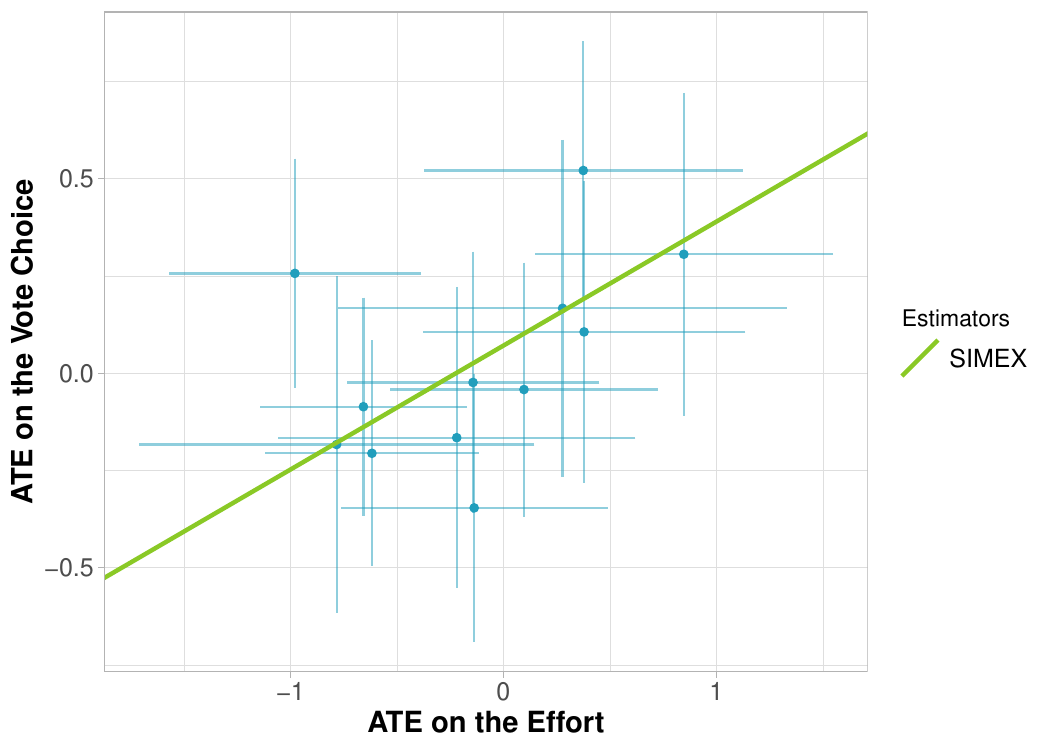}
    \caption{The figure plots the average treatment effects on the mediator (horizontal axis) and the outcome (vertical axis) across subgroups. Under a constant $\beta$, the dots are expected to lie on a straight line. The green line depicts the fitted $\beta$.}
    \label{fig:sub_use}
\end{figure}

Next, we estimate the average treatment effect on vote choice across various subgroups, focusing on the ``bad information'' arm, where the information reveals poor performance by the incumbent. These estimates are then used to calculate the indirect effect using SIMEX, treating them as aggregate-level data. Before that, it is helpful to plot the data and check for a clear linear relationship. If the data show linearity, it provides evidence that $\beta$ is constant, and we do not need to account for non-linearity. The final results are illustrated in Figure \ref{fig:sub_use}. We found a clear linear trend and that the estimated $\beta$ is 0.3, significant at the 0.1 level according to the SIMEX analysis. This suggests that the perception of high effort by the incumbent increases potential votes. The average causal mediation effect is 
$-0.08$, also significant at the 0.1 level. Consequently, we deduce that while the overall effect of information may not be substantial enough to detect in the field experiment, there is a significant indirect effect of bad information through voters' perceptions of politicians' effort. Specifically, bad information leads to fewer votes for the incumbent due to perceived lower effort. However, no significant results were detected for the ``good information'' arm, although the average mediation effect is positively aligned with our expectations.

\section{Discussion and Extension} \label{sec:ext}

Our identification strategy hinges on a pivotal assumption: $Cov(\gamma_k,\delta_k)=0$, implying that the treatment effect on the mediator is not correlated with other mechanisms. How can we assess this assumption in practice? The new identification strategy provides two opportunities to address the identification problem—during the design stage and the data analysis stage.

First, researchers can assess this assumption during the design stage, before any experiments or data collection. (1) This can be achieved through a theoretical examination. Understanding the underlying mechanisms necessitates a theoretical foundation. For instance, considering the voter turnout example in section \ref{sec:model}, theory identifies four major mechanisms that may influence a voter's decision to abstain or vote: relative payoff from the favored candidate, civic duty, the probability of being pivotal, and voting cost. If theory suggests these mechanisms are uncorrelated, then our assumption is valid. In the turnout example, most theories propose that these mechanisms represent \textit{distinct} aspects of voting calculus and lack evident correlation. For example, the probability of being pivotal is determined by the size of the population, whereas civic duty encompasses all moral considerations, which are unlikely to influence preferences for the candidate. These considerations are reflected in the formula $p_i b + D - C > 0$, where additivity implies a form of independence among the mediators. Given a clear treatment, if the mediators capture distinct mechanisms, the identification assumptions are more likely to hold. (2) If conducting an experiment, researchers can control the treatment to minimize potential correlations between $\gamma$ and $\delta$. For instance, when varying treatment intensity, it is advisable to limit the treatment elements to prevent correlations with other mechanisms.


Second, in the data analysis phase, we can leverage our novel decomposition, which takes the form $\tau_k=\mathbb{E} [\delta_k] + \beta \gamma_k +\epsilon_k$. This formulation highlights a close analogy to the traditional linear regression model. Should the assumption fail, variables denoted by $X$, moderating both the treatment effects on the mediator and other mechanisms, may exist. Thus, akin to the approach in traditional linear regression, we \textit{control} for $X$ to `purify' the error term. As our analysis is based on aggregate-level data (i.e., expected treatment effects), we need to account for $\mathbb{E}[X_k]$, the average of $X_k$ in the respective groups, in the linear regression model. More advanced and robust techniques, such as double machine learning, can also be employed. Cross-fitting is additionally recommended to reduce overfitting.


Owing to the challenge of obtaining a large number of observations of average treatment effects $(\gamma_k,\delta_k)$ in practice, indiscriminately adding all covariates in the linear regression model without careful consideration is not advised. This underscores the crucial role of theory in mediation analysis. The decision to include a covariate hinges on theoretical considerations regarding its potential to modify other mechanistic pathways. With theory, researchers can readily identify the most `important' omitted moderators, as measured by $R^2$, and incorporate them into the regression to enhance the robustness of their estimates \citep{oster2019unobservable}. With the extension to linear regression, researchers may also apply traditional \textit{sensitivity analysis} to evaluate the robustness of their estimates against unobserved confounders \citep[e.g.][]{cinelli2020making}.

\section{Conclusion}

Causal mediation analysis is inherently challenging. We hope our method provides an additional useful tool for applied researchers who seek to better understand causal mechanisms. In this study, we propose an alternative identification assumption and strategy that can enable researchers easily estimate causal mediation effects. The method converts the intricate mediation problem into a simple linear regression problem. Based on the isolated mechanism assumption, once researchers identify the treatment effects on the mediator and the outcome, our approach can consistently estimate the indirect effect. The proposed strategy enables researchers to address the identification problem in both the design and analysis stages.


While our method reduces the reliance on unconfoundedness assumptions, it may require greater data collection efforts to obtain precise estimates. Thankfully, in the era of big data, this should not be a major problem nowadays. Moreover, it is important to recognize that researchers should select appropriate mediation analysis techniques based on the specific design of their studies. All methods come with their own set of assumptions. There also remain several avenues for future research. Questions like post-treatment confounders, extending the method to non-binary treatments, identifying necessary assumptions for correlated mechanisms or integrating other causal identification strategies beyond instrumental variables for individual-level estimators remain open. Lastly, our method suggests a promising avenue to bridge causal mediation with causal moderation, indicating the potential for discovering other effective methodological combinations.

\newpage
\renewbibmacro{in:}{}
\printbibliography[]


\clearpage
\setcounter{page}{1}

\appendix
\addcontentsline{toc}{section}{Appendix} 
\part{Supplementary Information} 
\parttoc 

\setcounter{figure}{0}
\setcounter{table}{0}
\setcounter{definition}{0}
\setcounter{assumption}{0}
\renewcommand\thefigure{A.\arabic{figure}}
\renewcommand\thetable{A.\arabic{table}}
\renewcommand{\theassumption}{\Alph{section}\arabic{assumption}}

\clearpage

\section{More on Decomposition}\label{si:decomp}
It is evident that multiple versions of natural direct and indirect effects exist. The main concern is the interaction effect between the treatment and the mediator. For natural direct effect $\delta^i(t) = Y^i(1,M^i(t))-Y^i(0,M^i(t))$ and natural indirect effect $\eta^i(t)=Y^i(t,M^i(1))-Y^i(t,M^i(0))$, the value of $\delta^i(t)$ and $\eta^i(t)$ may depend on $t$, if there exists an interaction effect. In the binary treatment case ($t=1$ and $t'=0$), they have specific names \citep{robins1992identifiability}:

(1) The ``pure'' effect implies that no interaction effect is picked up. We call $\delta^i(0)= Y^i(1,M^i(0))-Y^i(0,M^i(0))$ the \textit{pure direct effect}, where the mediator is set to the value it would have been without treatment. Additionally, $\eta^i(0) = Y^i(0,M^i(1))-Y^i(0,M^i(0))$ is the \textit{pure indirect effect} where treatment is set to absent.

(2) The ``total'' effect captures the interaction effect. Therefore the \textit{total direct effect} is defined as $\delta^i(1)= Y^i(1,M^i(1))-Y^i(0,M^i(1))$, where the mediator takes the potential value if the treatment is on. Similarly, the \textit{total indirect effect} $\eta^i(1) = Y^i(1,M^i(1))-Y^i(1,M^i(0))$ set treatment to present.

Together, we obtain two different decompositions:
\begin{equation}\label{equ:twodecomp}
    \begin{aligned}
        \tau &= \delta(0)+\eta(1)\\
        \tau &= \delta(1)+\eta(0)
    \end{aligned}
\end{equation} Assuming no interaction effect exists between the treatment and mediator, the pure and total direct (indirect) effects should be the same because the effect does not depend on the mediator (treatment). Under the assumption, $\delta=\delta(0)=\delta(1)$ and $\eta=\eta(0)=\eta(1)$. Then the decomposition is unique: $\tau=\delta+\eta$.  

If people do not assume ``no interaction effect'', we can further decompose the total direct effect or total indirect effect to emphasize the interaction effect. \citet{vanderweele2013three} proposes further decomposing the total direct (or indirect) effect. For example, suppose the mediator is binary, then the total indirect effect can be decomposed into pure indirect effect and the interaction term: $\eta(1)= \eta(0) + [(Y(1,M(1))-Y(1,M(0))-Y(0,M(1))+Y(0,M(0))](M(1)-M(0))$.

Notably, other decomposition methods within the counterfactual approach exist. For example, the study by \citet{gallop2009mediation} examines mediation analysis under principal strata. However, as emphasized by \citet{vanderweele2011principal}, decomposition based on principal strata does not equate to the natural direct effect and natural indirect effect. Moreover, membership in principal strata themselves generally are unidentified. 

\section{Interaction Effect}\label{si:int}

As mentioned early, under the assumption of no interaction effect between the treatment and mediator, identification assumption \ref{ass:cov} is likely to hold in general. If we allow interaction effect, further justification may be required. However, actually, even if we allow interaction effect, if we are flexible to the parameter of interest (total or pure indirect effect), we can ``purify'' the interaction effect by using pure direct effect. Let me explain it by assuming a linear structural model with an interaction:
\begin{align}
        \mathbb{E}[Y] &= \alpha_1 + \delta T + \tilde{\beta} M + \theta TM \label{equ:inter1}\\
        \mathbb{E} [M] &= \alpha_2 + \gamma T  \label{equ:inter2}
\end{align} Here, parameter $\theta$ captures the interaction effect between the treatment and mediator.

Recall that with an interaction effect, there are two ways to decompose total causal effect, and it is not hard to show them under the above linear structural model:
\begin{equation}\label{equ:decom1}
\begin{aligned}
    \tau_1 &= \delta(1)+\eta(0) \\
    &=\mathbb{E}[Y^i(1,M^i(1))-Y^i(0,M^i(1))]+\mathbb{E}[Y^i(0,M(1))- Y^i(0,M(0))]\\
    &=[\delta+\theta(\alpha_2 +\gamma)]+(\tilde{\beta}\gamma)
\end{aligned}
\end{equation}

and
\begin{equation}\label{equ:decom2}
\begin{aligned}
    \tau_2 &= \delta(0)+\eta(1)\\
    &= \mathbb{E}[Y^i(1,M^i(0))-Y^i(0,M^i(0))]+\mathbb{E}[Y^i(1,M(1))- Y^i(1,M(0))]\\
    &=(\delta+\theta \alpha_2) + [(\tilde{\beta}+\theta)\gamma]
\end{aligned}
\end{equation}

It needs to be noted that the total effect has a unique representation $\delta+\theta \alpha_2 + \tilde{\beta}\gamma+\theta \gamma$, i.e., $\tau_1=\tau_2$. Therefore, relationship between $\tau$ and $\gamma$ is unique. However, it has two interpretations $\delta(0)+\eta(1)$ and $\delta(1)+\eta(0)$ by considering different components at one time. 

Corresponding to representation \eqref{equ:lm3}, in the decomposition \eqref{equ:decom1}, the ``parameter'' $\beta$ is equal to $\tilde{\beta}$. We note that the total direct effect $\delta(1)$ contains $\gamma$. However, in the decomposition \eqref{equ:decom2}, the parameter $\beta=\tilde{\beta}+\theta$ and $\delta(0)$ does not contain $\gamma$. Obviously, in the latter decomposition, it is easy to have $Cov(\gamma_k,\delta_k)=0$. Then, we can use OLS to estimate $\tilde{\beta}+\theta$ and the total indirect effect $\eta(1)$ in \eqref{equ:decom2}. However, in general, we cannot consistently estimate $\tilde{\beta}$ and thus pure indirect effect $\eta(0)$ with \eqref{equ:decom1} because assumption \ref{ass:cov} generally does not hold in this decomposition. Therefore, we should use the first decomposition in practice.

This example also highlights the interpretation of $\beta$, which is similar to the relationship between reduced-from and structural model. Although we estimate the same reduced-form model $\tau_k=\mathbb{E}[\delta_k]+\beta\gamma +\varepsilon_k$, under different assumptions, $\beta$ represents different structural parameters. If there exists the interaction effect, $\beta$ has two parts, one is $\tilde{\beta}$ (the effect of the mediator on the outcome), and the second part is $\theta$ (the interaction effect).

\section{Multiple Mechanisms} \label{app:mul}
In this section, we consider the general case allowing multiple independent mechanisms. We let $M_1$ be the mediator of interests, and let $M_{-1}$ be other mediators. We define the overall effect by $\tau^i$:

$$
\begin{aligned}
\tau^i=  Y^i(1,M^i_1(1),...,M^i_j(1))- Y^i(0,M^i_1(0),...,M^i_j(0)).
\end{aligned}
$$

The direct effect is defined as

$$
\begin{aligned}
\delta^i(t') = Y^i(t', M^i_1(t'), ..., M^i_J(t'))- Y^i(t, M^i_1(t'),...,M^i_J(t')).
\end{aligned}
$$

Because we allow multiple mechanisms, apart from the convention, we use $j-$ and $j+$ to denote index $h \in J$ such that $h<j$ and $h>j$, respectively. For the indirect effect, we define as

\begin{equation*}
	\begin{aligned}
\eta_j^i(t', t) = Y^i&(t, M_{j-}(t) ,M_j(t'),M_{j+}(t'))- \\ 
&Y^i(t, M_{j-}(t) ,M_j(t),M_{j+}(t'))
	\end{aligned}
\end{equation*}

The overall causal effect can be decomposed as:

$$
\begin{aligned}
    \tau^i = \delta^i(t')+\sum_{j = 2}^J \eta_j^{i}(t',t)+\eta_1^{i}(t',t)
\end{aligned}
$$

To verify it, we let $t'=1$ and $t=0$;
$$
    \begin{aligned}
        \tau^i &= \delta^i(1)+\sum_{j = 1}^J \eta_j^{i}(1,0)\\
        &= Y^i(1, M^i_1(1), ..., M^i_J(1))- Y^i(0, M^i_1(1),...,M^i_J(1)) \\
        &+Y^i(0, M_1(1),...,M_j(1),)-Y^i(0, M_1(0),... ,M_j(1))\\
        &+Y^i(0, M_1(0),M_2(1)...,M_j(1),)-Y^i(0, M_1(0),M_2(0),... ,M_j(1))\\
        &+...\\
        &=Y^i(1, M^i_1(1), ..., M^i_J(1))-Y^i(0, M^i_1(0), ..., M^i_J(0))
    \end{aligned}
$$ Basically, the first term in each line is canceled out by the second term in the previous line.

Notably, previous definitions are not general enough. For example, in the direct effect, we require all mediators to take potential outcomes under treatment $t'$. In general, different mediators can take different potential outcomes. Similarly, for the indirect effect $\eta_j$, different mediators other than $j$ can take any possible potential outcomes. But whatever potential outcomes they take, our results hold if the mechanism of interests is additively separable from other mechanisms:

\begin{align}
    \tau^i = (\delta^i + \sum_{j = 2}^J \eta_j^{i})+\eta_1^{i}
\end{align} The average level decomposition has the similar form: $\tau = (\delta + \sum_{j = 2}^J \eta_j)+\eta_1$.

\clearpage

\section{Other Identification Strategies}\label{si:ids}

The identification assumptions in main test section derive basic requirements for mediation analysis; however, they do not tell us how to satisfy those assumptions. In other words, we need identification strategies \citep{angrist1999empirical,samii2016causal}. As one of the most widely used econometric tools, instrumental variables (IV) has been proposed to help identify indirect effects. \citet{frolich2017direct} consider two independent IVs, one for the treatment and the other for the mediator. With two IVs, they specify required assumptions and propose estimators for non-parametric identification of the indirect effect. \citet{rudolph2021causal} extend the results to two related IVs. They also consider the case of a single IV for the treatment. Unfortunately, in this case, they conclude that we still need to rely on the assumption of no unobserved confounders of the mediator-outcome relationship.

\begin{figure}
   \begin{center}
       \begin{tikzpicture}
           \node at (0,0) (t) {$T$};
           \node at (2,0) (m) {$M$};
           \node at (4,0) (y) {$Y$};   
           \node[blue] at (1,-1) (u1) {$U_1$};
            \node[red] at (3,-1) (u2) {$U_2$};
            \draw[->,red] (u2) -- (m);
           \draw[->,red] (u2) -- (y);
           \draw[->,blue] (u1) -- (m);
           \draw[->,blue] (u1) -- (t);
        \path[dashed,<->]
            (u1) edge [bend right=45] (u2);
           
          \path[->] (t) edge node[auto] {$\gamma$} (m);
          \path[->] (m) edge node[auto] {$\tilde{\beta}$} (y);
           \path[->]
            (t) edge [bend left=45] node[auto] {$\delta$} (y);   
       \end{tikzpicture}
   \end{center}
   \caption{Mediation analysis generally requires addressing confounders $U_1$ and $U_2.$ }\label{fig:dag2}
\end{figure}

By observing Figure \ref{fig:dag2}, an interesting idea to address $U_2$ is to treat the treatment $T$ as an IV for the mediator $M$. Therefore, we only need one IV for the treatment (or treatment is randomly assigned). To be a valid IV, however, it is well-known that $T$ cannot have a direct effect on $Y$ except from $M$. \citet{sobel2008identification} explores the identification of indirect effects under this exclusion assumption. It is clear that $\delta=0$ is not a typical case. To account for the violation, \citet{strezhnev2021testing} develops a useful sensitivity analysis method. \citet{small2011mediation} proposes a different IV method to bypass the exclusion assumption; however, it requires the interaction between covariate $X$ and a randomly assigned treatment $T$ to be a valid IV for $M$. Recently, \citet{dippel2019mediation} found a new assumption allowing us only to use one IV for the treatment. The assumption adds constraints on the distribution of the unobserved confounding variables: unobserved confounding variables that jointly cause the treatment and the intermediate outcome are independent of the confounders that cause the intermediate and the final outcome, that is, $U_1$ and $U_2$ are independent.

Because most IV methods are developed under a linear structural model, one important feature is that they require a kind of ``constant'' effect assumption or zero covariance assumption ($Cov(\beta^i,\gamma^i)$) we mentioned before. \footnote{See more detailed discussion by \citet{hong2015causality}.} Other identification strategies exist, for example, experimental designs \citep{acharya2018analyzing}, modified Difference-in-differences (DID), and synthetic control method (SCM) (See the survey by \citet{celli2022causal}).

In general, current identification strategies for mediation analysis still depend on several strong assumptions. Moreover, most are designed for IVs, which constrains the scope of application because it is not easy to find a good IV in many studies. 


\clearpage

\section{Illustration of Attenuation Bias}\label{si:atten}

Most estimators in applied research are asymptotically normal. Therefore, without loss of generality, we also assume our estimates $(\hat{\tau}_k,\hat{\gamma}_k)$ are normally distributed around the true values $(\tau_k,\gamma_k)$.\footnote{This assumption is more plausible when the sample size is large for each group. In the era of big data, this should not be a major problem.}

\begin{assumption}[Heterogeneous Measurement]\label{ass:mea}
    \begin{align}
    \hat{\gamma}_k &= \gamma_k + u_k\\
     \hat{\tau}_k &= \tau_k + v_k
\end{align} where $u_k \sim N(0,\sigma_{uk}^2)$ and $v_k \sim N(0,\sigma_{vk}^2)$, $Cov(\gamma_k,u_k)=0$, $Cov(\gamma_k,v_k)=0$, $\sigma_{uk}^2>0$, and $\sigma_{vk}^2>0$.
\end{assumption} 

In the assumption, as the classical setting, we also assume $Cov(\gamma_k,u_k)=0$, and $Cov(\gamma_k,v_k)=0$. However, departing from the CEV, we allow each estimate to have its own variance $\sigma^2_{\cdot k}$. This is more general, and more realistic because it is implausible that treatment effect has the same asymptotic variance across subgroups. Notably, our estimators introduced later are robust to the correlation between $u_k$ and $v_k$. The following proposition illustrates the attenuation under our mediation analysis framework and proposes an adjusted estimator.

\begin{proposition}\label{prop:meta}
Suppose $(\tau_k,\gamma_k)$ satisfies the decomposition \eqref{equ:lm3}: $\tau_k = \mathbb{E}[\delta_k]+\beta_k\gamma_k+\varepsilon_k$, and the observed random sample $(\hat{\tau}_k,\hat{\gamma}_k)$ follows the measurement assumption \ref{ass:mea}.

Let $\sigma^2_\gamma$ be $Var(\gamma_k)$ and  $\lambda = \frac{\sigma^2_\gamma}{\sigma^2_\gamma+\overline{\sigma^2_{uk}}}$. Considering the estimator $\hat{\beta}=\frac{\sum_{k=1}^K(\hat{\gamma}_k-\overline{\hat{\gamma}_k})\hat{\tau}_k}{\sum_{k=1}^K(\hat{\gamma}_k-\overline{\hat{\gamma}_k})^2}$, under assumption $\sum_{k=1}^\infty \frac{Var(\hat{\gamma}_k^2)}{k^2}<\infty$ and assumption \ref{ass:cov},

(1) If $\beta$ is a constant, then $\lambda^{-1} \hat{\beta} \xrightarrow{p} \beta$ as $K \rightarrow \infty$;

(2) If $\beta_k$ is a random variable, then $\lambda^{-1} \hat{\beta} \xrightarrow{p} \mathbb{E} [\beta_k]$ as $K \rightarrow \infty$ under assumption $\beta_k\independent \gamma_k$.


\end{proposition}

In the above proposition \ref{prop:meta}, we also assume $\sum_{k=1}^\infty \frac{Var(\hat{\gamma}^2_k)}{k^2}<\infty$. This technical assumption is required because, in the proof, we apply Kolmogorov’s strong law of large numbers with independent but not identically distributed samples. The proposition suggests using $\lambda^{-1}\hat{\beta}$ as a consistent estimator. For $\lambda=\frac{\sigma^2_\gamma}{\sigma^2_\gamma+\overline{\sigma^2_{uk}}}$, the numerator is the variance of true $\gamma_k$; in the denominator, $\overline{\sigma^2_{uk}}$ is the mean of the variance of $\hat{\gamma}_k$. Because the denominator is always larger than the numerator, $\lambda<1$. In practice, we have data $\sigma^2_{uk}$ and therefore can calculate the sample average $\overline{\sigma^2_{uk}}$. However, we need an estimate of the unknown $\sigma^2_\gamma$, the variance of the true $\gamma$. The variance can be regarded as the ``inter-study variance'' in the random-effects model \footnote{In the random-effects model, observed treatment effect $y_i$ is assumed to be a function of the true treatment effect for the study $\theta_i$ and the sampling error $e_i$: $y_i=\theta_i+e_i$; and $\theta_i$ can be decomposed as $\mu+\delta_i$ where $\mu$ is the overall treatment effect and $\delta_i$ is the deviation of the $i$'s-study's effect from the overall effect. The variance of $\delta_i$ is the inter-study variance. If, in the special case, it equals 0, we have the fixed-effect model.}. Many estimators in the meta-analysis literature exist \citep{dersimonian1986meta,paule1982consensus,viechtbauer2005bias,dersimonian2007random}.

\section{Proofs}

\subsection{Proof of Proposition \ref{prop:consist}}\label{sec:app_p1}

\begin{proof}
We first decompose the average total causal effect $\tau$ as follows:
$$
\begin{aligned}
    \tau(t,t') &= \mathbb{E}[Y^i(t,M^i(t))- Y^i(t',M^i(t))]+ \mathbb{E} [Y^i(t',M^i(t))-Y^i(t',M^i(t'))]\\
    &=\mathbb{E}[Y^i(t,M^i(t))- Y^i(t',M^i(t))] + \frac{\mathbb{E} [Y^i(t',M^i(t))-Y^i(t',M^i(t'))]}{\mathbb{E}[M^i(t)-M^i(t')]} \times\mathbb{E}[M^i(t)-M^i(t')] \\
    &:= \delta + \beta \gamma 
    \end{aligned}
$$
Then, given the random sample $(\tau_k,\delta_k,\beta_k,\gamma_k)$, we convert it to be simple linear regression
\begin{align}
    \tau_k = \mathbb{E}\delta_k+\beta_k\gamma_k + \varepsilon_k
\end{align} where $\varepsilon_k = \delta_k - \mathbb{E}\delta_k$.


Consider the estimator 
$\hat{\beta}=\frac{\sum_{k=1}^K(\gamma_k-\overline{\gamma}_k)\tau_k}{\sum_{k=1}^K(\gamma_k-\overline{\gamma}_k)^2}$.


We first to show result (2) that $\hat{\beta} \rightarrow \mathbb{E}\beta_k$. Note that 

(1) By construction, $\mathbb{E}\varepsilon_k=\mathbb{E}\delta_k-\mathbb{E}\delta_k=0$;

(2) Assumption 2 implies that $\mathbb{E}[\gamma_k\epsilon_k]=\mathbb{E}[\gamma_k(\delta_k-\mathbb{E}\delta_k)]=Cov(\gamma_k,\delta_k)=0$. In the proof, I also use $\overline{\delta}_k$ to denote $\mathbb{E}\delta_k$ in order to avoid confusion with other expectation notation.
\begin{align}
\hat{\beta}&=\frac{\sum_{k=1}^K(\gamma_k-\overline{\gamma}_k)\tau_k}{\sum_{k=1}^K(\gamma_k-\overline{\gamma}_k)^2}\\
&=\frac{\sum_{k=1}^K(\gamma_k-\overline{\gamma}_k)(\overline{\delta}_k+\beta_k\gamma_k + \varepsilon_k)}{\sum_{k=1}^K(\gamma_k-\overline{\gamma}_k)^2}\\
&=\frac{ \frac{1}{K} \sum_{k=1}^K(\gamma_k-\overline{\gamma}_k)\beta_k\gamma_k}{\frac{1}{K} \sum_{k=1}^K(\gamma_k-\overline{\gamma}_k)^2}+\frac{\frac{1}{K} \sum_{k=1}^K(\gamma_k-\overline{\gamma}_k)\varepsilon_k}{\frac{1}{K} \sum_{k=1}^K(\gamma_k-\overline{\gamma}_k)^2}\label{equ:pf1}\\
& \xrightarrow{p} \frac{\mathbb{E}[\gamma^2_k\beta_k] -\mathbb{E}[\gamma_k] \mathbb{E}[\gamma_k\beta_k]}{Var(\gamma_k)} + \frac{\mathbb{E}[(\gamma_k-\overline{\gamma}_k)\varepsilon_k]}{Var(\gamma_k)}\label{equ:pf2}\\
&= \frac{Var(\gamma_k)\mathbb{E}\beta_k}{Var(\gamma_k)} + \frac{\mathbb{E}\gamma_k \varepsilon_k}{Var(\gamma_k)} \label{equ:pf3}\\
&= \mathbb{E}\beta_k
\end{align} where line $\eqref{equ:pf1}$ comes from $\overline{\delta}_k \sum_{k=1}^K(\gamma_k-\overline{\gamma}_k)=0$, line $\eqref{equ:pf2}$ is implied by Slutsky's Lemma, $\eqref{equ:pf3}$ is implied by 
$\mathbb{E}\varepsilon_k$=0 and assumption $\beta_k \independent \gamma_k$, the last line is implied by $\mathbb{E}[\gamma_k\epsilon_k]=0$, 


Result (1) trivially follows the same logic.

\end{proof}

\subsection{Proposition \ref{prop:consist} under Multiple Mechanisms}\label{app:ext}

In the main text, when we discuss our novel decomposition and identification assumptions, we consider ``no interaction effect'' so that the decomposition is unique. Here, we can slightly relax it to be ``no interaction effect with respect to $M_1$''. 

\begin{assumption}[No interaction effect with respect to $M_1$] \label{ass:app_1}
For any $t_j \in \{0,1\}$ where $j=1,2,..,J$,
    $\eta_1(t_1,M_1(1),M_2(t_2),...,M_2(t_j))-\eta_1(t_1,M_1(0),M_2(t_2),...,M_2(t_j))=B$ 
\end{assumption}
In other words, the assumption allows any possible interaction effect among the treatment $T$ and other mediators $M_{-1}$; however, the effect of $M_1$ does not depend on them. Under this assumption, without loss of the generality, we use $\Delta$ to denote $\delta + \sum_{j = 2}^J \eta_j$ and thus 
\begin{align}
    \tau = \Delta + \eta_1 \label{equ:app_decom}
\end{align}

Similarly, to have a unique form of $\gamma$ (the treatment effect on the mediator $M_1$), we also need a kind of ``no interaction effect.''

\begin{assumption}[No interaction effect] \label{ass:app_2}

For any $t_j \in \{0,1\}$ where $j=2,..,J$,
    $$\mathbb{E}Y^i(M_1(1),M_2(t_2),...,M_2(t_j))-\mathbb{E}Y^i(M_1(0),M_2(t_2),...,M_2(t_j))=D$$
\end{assumption}

Under the above two ``no interaction effect'' assumptions, subsequently, we can modify the Proposition \ref{prop:consist} as follows:

\begin{proposition}
Let $(\tau,\Delta,\gamma)$ are random variables. Given the random sample $(\tau_k,\gamma_k)_{k \in K}$. Suppose $Var(\gamma_k)>0$ and $Cov(\gamma_k,\Delta_k)=0$.

Considering the estimator $\hat{\beta}=\frac{\sum_{k=1}^K(\gamma_k-\overline{\gamma}_k)\tau_k}{\sum_{k=1}^K(\gamma_k-\overline{\gamma}_k)^2}$.

(1) If $\beta$ is a constant, then $\hat{\beta} \xrightarrow{p} \beta$ as $K \rightarrow \infty$;

(2) If $\beta_k$ is a random variable, then $\hat{\beta} \xrightarrow{p} \mathbb{E} \beta_k$ as $K \rightarrow \infty$ under assumption $\beta_k\perp \gamma_k$

and thus $\eta_k$ is consistently estimated by $\hat{\beta}\gamma_k$.
\end{proposition} 

The key difference between the above-modified proposition and the original one is the identification assumption. Here, we need $Cov(\gamma_k,\Delta_k)=0$. It means that the treatment effect on the mediator of interest is not correlated to the direct effect and other mechanisms.

\subsection{Proof of Unbiasedness }\label{app:unbias}

\begin{proposition}\label{prop:unbias}
Let $(\tau,\delta,\gamma)$ be random variables and as defined in \eqref{equ:p2} and \eqref{equ:p3}. Given the random sample $(\tau_k,\gamma_k)_{k \in K}$. Suppose following two assumptions hold:

(1) (Variance) $Var(\gamma_k)>0$;

(2) (Mean Independence) $\mathbb{E}[\delta_k|\gamma_k]=\mathbb{E}[\delta_k]$

Considering the estimator $\hat{\beta}=\frac{\sum_{k=1}^K(\gamma_k-\overline{\gamma}_k)\tau_k}{\sum_{k=1}^K(\gamma_k-\overline{\gamma}_k)^2}$.

(1) If $\beta$ is a constant, then $\mathbb{E}\hat{\beta}$ = $\beta$;

(2) If $\beta_k$ is a random variable, then $\mathbb{E}\hat{\beta} = \mathbb{E}\beta_k$ under assumption $\mathbb{E}[\beta_k|\gamma_k]=\mathbb{E}\beta_k$,

and thus $\eta_k$ is unbiased.

\end{proposition}

\begin{proof}

For unbiasedness, note that by construction, $\mathbb{E}\varepsilon_k=\mathbb{E}\delta_k-\mathbb{E}\delta_k=0$ and thus with mean independence assumption (2) we have $\mathbb{E}[\epsilon_k|\gamma_k]=\mathbb{E}\varepsilon_k=0$. From line \eqref{equ:pf1}, we take the expectation given observed $\gamma_1,\gamma_2,...,\gamma_K$,
\begin{align}
    \mathbb{E}[\hat{\beta}|\gamma_1,\gamma_2,...,\gamma_K] 
    &= \mathbb{E}[\beta_k] \frac{  \sum_{k=1}^K(\gamma_k-\overline{\gamma}_k)\gamma_k}{\sum_{k=1}^K(\gamma_k-\overline{\gamma}_k)^2} + \frac{\sum_{k=1}^K(\gamma_k-\overline{\gamma}_k)\mathbb{E}[\epsilon_k|\gamma_k]}{\sum_{k=1}^K(\gamma_k-\overline{\gamma}_k)^2}\\
    &=\mathbb{E}\beta_k
\end{align}

Result (1) trivially follows the same logic.

\end{proof}

\subsection{BCES estimator}\label{sec:bces}

Ideally, if we have data on the true value $(\gamma_k,\tau_k)$, the OLS estimator is consistent, from Proposition \ref{prop:consist} and the proof \ref{sec:app_p1}:
\begin{align}  \hat{\beta}_{ideal}&=\frac{\sum_{k=1}^K(\gamma_k-\overline{\gamma}_k)\tau_k}{\sum_{k=1}^K(\gamma_k-\overline{\gamma}_k)^2}\\
& \rightarrow \frac{\sigma^2_{\gamma} \mathbb{E}\beta_k}{\sigma^2_{\gamma}}
\end{align}

However, we only observe $(\hat{\gamma}_k,\hat{\tau}_k)$; therefore, the empirical estimator is attenuated, by proof \ref{pf:meta}:
\begin{align}
\hat{\beta}&=\frac{\sum_{k=1}^K(\hat{\gamma}_k-\overline{\hat{\gamma}_k})\hat{\tau}_k}{\sum_{k=1}^K(\hat{\gamma}_k-\overline{\hat{\gamma}_k})^2}\\
& = \frac{\frac{1}{K} \sum_{k=1}^K(\hat{\gamma}_k-\overline{\hat{\gamma}_k})[\mathbb{E}\delta + \gamma_k \beta_k + (\varepsilon_k+v_k)]}{\frac{1}{K} \sum_{k=1}^K(\hat{\gamma}_k-\overline{\hat{\gamma}_k})^2}\\
& \rightarrow \frac{\sigma^2_{\gamma} \mathbb{E}\beta_k}{\sigma^2_{\gamma}+\overline{\sigma^2_{uk}}}
\end{align}

Therefore, to obtain a consistent estimator, in the denominator, we could subtract $\overline{\sigma^2_{uk}}$. The modified estimator is exactly the BCES estimator:
\begin{align}  \hat{\beta}_{BCES}=\frac{\sum_{k=1}^K(\hat{\gamma}_k-\overline{\hat{\gamma}_k})\hat{\tau}_k}{\sum_{k=1}^K(\hat{\gamma}_k-\overline{\hat{\gamma}_k})^2-\sum_{k=1}^K\sigma^2_{uk}}
\end{align}

If we allow correlation between $u_k$ and $v_k$, we should adjust the numerator as well. Let $\sigma^2_{uvk}$ denote the covariance for observation $k$. The resulting BCES estimator is the same as the one proposed in the \cite{akritas1996linear}:
\begin{align}  \hat{\beta}_{BCES}=\frac{\sum_{k=1}^K(\hat{\gamma}_k-\overline{\hat{\gamma}_k})\hat{\tau}_k-\sum_{k=1}^K \sigma^2_{uvk}}{\sum_{k=1}^K(\hat{\gamma}_k-\overline{\hat{\gamma}_k})^2-\sum_{k=1}^K\sigma^2_{uk}}
\end{align}

\subsection{Proof of Proposition \ref{prop:twostep_ols}}

\begin{proof}
Let $Z_k = (1,\gamma_k)'$, $\widehat Z_k = (1,\widehat\gamma_k)'$, and $\theta_0 = (\alpha_0,\beta)'$.
The population second-stage model can be written as $\tau_k = Z_k'\theta_0 + \varepsilon_k.$
The isolated-mechanism condition implies $
    \mathbb E[Z_k\varepsilon_k]=0$,
because
\[
    \mathbb E[\varepsilon_k]=0,
    \qquad
    \mathbb E[\gamma_k\varepsilon_k]
    =
    \mathbb E[(\gamma_k-\mu_\gamma)\varepsilon_k]
    +
    \mu_\gamma\mathbb E[\varepsilon_k]
    =
    0.
\]

First consider the infeasible OLS estimator that uses the true group-level effects:
\[
    \widetilde\beta_K
    =
    \frac{
        \sum_{k=1}^K
        (\gamma_k-\bar\gamma)(\tau_k-\bar\tau)
    }{
        \sum_{k=1}^K
        (\gamma_k-\bar\gamma)^2
    },
\]
where $\bar\gamma = \frac{1}{K}\sum_{k=1}^K \gamma_k$,$\bar\tau = \frac{1}{K}\sum_{k=1}^K \tau_k$. Since $\tau_k = \mathbb{E}[\delta_k] + \beta\gamma_k + \varepsilon_k$, we have
\[
    \widetilde\beta_K-\beta
    =
    \frac{
        K^{-1}\sum_{k=1}^K(\gamma_k-\bar\gamma)\varepsilon_k
    }{
        K^{-1}\sum_{k=1}^K(\gamma_k-\bar\gamma)^2
    }.
\]
Standard OLS theorem tells us
\[
    \sqrt K(\widetilde\beta_K-\beta)
    \xrightarrow{d}
    \mathcal N\left(0,\frac{\Omega}{Q^2}\right).
\]

It remains to show that replacing \((\gamma_k,\tau_k)\) with
\((\widehat\gamma_k,\widehat\tau_k)\) does not change the first-order limit. Define
\[
    A_K = \frac{1}{K}\sum_{k=1}^K Z_kZ_k',
    \qquad
    b_K = \frac{1}{K}\sum_{k=1}^K Z_k\tau_k,
\]
and
\[
    \widehat A_K
    =
    \frac{1}{K}\sum_{k=1}^K \widehat Z_k\widehat Z_k',
    \qquad
    \widehat b_K
    =
    \frac{1}{K}\sum_{k=1}^K \widehat Z_k\widehat\tau_k.
\]
Let
\[
    \widetilde\theta_K=A_K^{-1}b_K,
    \qquad
    \widehat\theta_K=\widehat A_K^{-1}\widehat b_K.
\]
The slope components of \(\widetilde\theta_K\) and \(\widehat\theta_K\) are
\(\widetilde\beta_K\) and \(\widehat\beta_K\), respectively.

Since $\widehat\gamma_k=\gamma_k+u_{k}$,$\widehat\tau_k=\tau_k+v_{k}$, we have
\[
    \widehat A_K-A_K
    =
    \frac{1}{K}
    \sum_{k=1}^K
    \begin{pmatrix}
        0 & u_{k}\\
        u_{k} & 2\gamma_k u_{k}+u_{k}^2
    \end{pmatrix},
\]
and
\[
    \widehat b_K-b_K
    =
    \frac{1}{K}
    \sum_{k=1}^K
    \begin{pmatrix}
        v_{k}\\
        \gamma_k v_{k}+\tau_k u_{k}+u_{k}v_{k}
    \end{pmatrix}.
\]


so
\[
    \sum_{k=1}^K(u_{k}^2+v_{k}^2)=o_p(1).
\]
By Cauchy's inequality,
\[
    \frac{1}{\sqrt K}\sum_{k=1}^K |u_{k}|
    \leq
    \left(\sum_{k=1}^K u_{k}^2\right)^{1/2}
    =
    o_p(1),
\]
and
\[
    \frac{1}{\sqrt K}\sum_{k=1}^K |\gamma_k u_{k}|
    \leq
    \left(\frac{1}{K}\sum_{k=1}^K\gamma_k^2\right)^{1/2}
    \left(\sum_{k=1}^K u_{k}^2\right)^{1/2}
    =
    o_p(1).
\]
The same bounds apply to all components of
\(\widehat A_K-A_K\) and \(\widehat b_K-b_K\). Hence
$\sqrt K(\widehat A_K-A_K)=o_p(1)$, and $\sqrt K(\widehat b_K-b_K)=o_p(1)$. Using the identity
\[
    \widehat\theta_K-\widetilde\theta_K
    =
    \widehat A_K^{-1}
    \left[
        (\widehat b_K-b_K)
        -
        (\widehat A_K-A_K)\widetilde\theta_K
    \right],
\]
and noting that \(\widehat A_K^{-1}=O_p(1)\) and
\(\widetilde\theta_K=O_p(1)\), we obtain
$ \sqrt K(\widehat\theta_K-\widetilde\theta_K)=o_p(1)$.
Therefore the slope components satisfy
\[
    \sqrt K(\widehat\beta_K-\widetilde\beta_K)=o_p(1).
\]
Combining this result with the asymptotic normality of
\(\widetilde\beta_K\) gives
\[
    \sqrt K(\widehat\beta_K-\beta)
    =
    \frac{1}{Q}
    \frac{1}{\sqrt K}
    \sum_{k=1}^K
    (\gamma_k-\mu_\gamma)\varepsilon_k
    + o_p(1),
\]
and hence
\[
    \sqrt K(\widehat\beta_K-\beta)
    \xrightarrow{d}
    \mathcal N\left(0,\frac{\Omega}{Q^2}\right).
\]
This proves the theorem.
\end{proof}

\subsection{Proof of Proposition \ref{prop:inference}}

\begin{proof}
We start from simplifying some notations. The overall average indirect effect is $\mathbb{E}[\eta_k]=\sum_{k=1}^K\mathbb{E}[\beta_k\gamma_k] \mathbb{P}[\gamma_k]$. We use $\lambda_k$ to  simplify the notation $\mathbb{P}[\gamma_k]$, the proportion of $\gamma_k$ in the population, and then we use $\gamma_0$ to denote $\sum_{k=1}^K\gamma_k \lambda_k$. Under simplified notation, the overall average indirect effect is $\mathbb{E}[\beta_k]\gamma_0$. It is estimated by $\hat{\beta} \hat{\gamma}_0:=\hat{\beta}\sum_{k=1}^K \hat{\gamma}_k \hat{\lambda}_k$, where we let $\hat{\lambda}_k$ to denote the consistent estimate of the proportion $\lambda_k$, and $\hat{\gamma}_0$ denotes $\sum_{k=1}^K\hat{\gamma}_k \hat{\lambda}_k$.

 The null hypothesis we consider is $H_0: \mathbb{E}[\beta_k]\gamma_0=0$, which can be expressed as the union, $H_0: \mathbb{E}[\beta_k]=0 \cup \gamma_0=0$. Therefore, $H_0$ is rejected if we reject both parts of the null hypothesis. It is known that the (asymptotic) level $\alpha$ test is given by the (asymptotic) level $\alpha$ test of both parts \citep[for example see ][theorem 8.3.23]{berger2001statistical}.

    To test whether $\mathbb{E}[\beta_k]=0$, we apply the asymptotic normality of $\hat{\beta}$. Therefore, as traditional test, the $p$-value is $2*\Phi_Z(-|\frac{\hat{\beta}}{\hat{\sigma}_{\beta}}|)$, which is $p_{\beta}$.
    
    To test $\gamma_0=0$, given normality assumption, we know $\hat{\gamma}_{0} \rightarrow N(\gamma_0, \sum_{k=1}^K \lambda^2 \sigma^2_{uk})$. Therefore, as traditional test, the $p$-value is $2*\Phi_Z(-|\frac{\hat{\gamma}_{0}}{\sqrt{\sum_{k=1}^K \hat{\lambda}^2 \hat{\sigma}^2_{uk}}}|)$, which is $p_{\gamma}$. It follows that rejection of the Null hypothesis if $p_{\beta}\leq \alpha$ and $p_{\gamma}\leq \alpha$ is an asymptotic level $\alpha$ test.
\end{proof}

\subsection{Proof of Proposition \ref{prop:inference2}}

\begin{proof}
    
     Let $\Phi_Z(\cdot):=\Phi_Z (\frac{1+\sqrt{1-\alpha}}{2})$. From the proof of \ref{prop:inference}, we can construct the separate asymptotic $\sqrt{(1-\alpha)}\%$ confidence interval as 
    $$
    \begin{aligned}
        \mathbb{P}[\hat{\gamma}_0 - \Phi^{-1}_Z (\cdot)* \sqrt{\sum_{k=1}^K \hat{\lambda}^2 \hat{\sigma}^2_{uk}} \leq \gamma_0 \leq \hat{\gamma}_0 + \Phi^{-1}_Z (\cdot)* \sqrt{\sum_{k=1}^K \hat{\lambda}^2 \hat{\sigma}^2_{uk}} ] &= \sqrt{(1-\alpha)}\\
         \mathbb{P}[\hat{\beta} - \Phi^{-1}_Z (\cdot)* \hat{\sigma}_{\beta} \leq \mathbb{E}[\beta_k] \leq \hat{\beta} + \Phi^{-1}_Z (\cdot)* \hat{\sigma}_{\beta}] &= \sqrt{(1-\alpha)}
    \end{aligned}
    $$

    that is,

    $$
    \begin{aligned}
        \mathbb{P}[a_1 \leq \gamma_0 \leq a_2 ] &= \sqrt{(1-\alpha)}\\
         \mathbb{P}[a_3 \leq \mathbb{E}[\beta_k] \leq  a_4] &= \sqrt{(1-\alpha)}
    \end{aligned}
    $$

Because $a_1 \leq \gamma_0 \leq a_2 $ and $a_3 \leq \mathbb{E}[\beta_k] \leq a_4$ implies $\min[a_j] \leq \mathbb{E}[\beta_k]\gamma_0 \leq \max[a_j], j\in \{1,2,3,4\}$, we get 

$$
\mathbb{P}[\underline{a} \leq \mathbb{E}[\beta_k]\gamma_0\leq \overline{a}] \ge 1-\alpha
$$
    
\end{proof}

\subsection{Proof of Proposition \ref{prop:meta}}\label{pf:meta}

\begin{proof}

Firstly, We calculate the expectation of $\hat{\gamma}_k^2=\gamma_k^2+2\gamma_k u_k +u^2_k$. Let $\mu_\gamma = \mathbb{E}\gamma_k$.

For each part, we have
\begin{align}
&\mathbb{E}\gamma_k^2=\sigma^2_\gamma+\mu^2_\gamma\\
   & \mathbb{E}2\gamma_k u_k=0 \text{  by  } Cov(\gamma_k,u_k)=0 \\
  & \mathbb{E} u^2_k= \sigma^2_{uk}
\end{align}

Therefore, $\mathbb{E}\hat{\gamma}_k^2=\sigma^2_\gamma+\mu^2_\gamma+\sigma^2_{uk}$.

Now, considering the estimator,
\begin{align}
\hat{\beta}&=\frac{\sum_{k=1}^K(\hat{\gamma}_k-\overline{\hat{\gamma}_k})\hat{\tau}_k}{\sum_{k=1}^K(\hat{\gamma}_k-\overline{\hat{\gamma}_k})^2}\\
& = \frac{\frac{1}{K} \sum_{k=1}^K(\hat{\gamma}_k-\overline{\hat{\gamma}_k})[\mathbb{E}\delta_k + \gamma_k \beta_k + (\varepsilon_k+v_k)]}{\frac{1}{K} \sum_{k=1}^K(\hat{\gamma}_k-\overline{\hat{\gamma}_k})^2}
\end{align}

To see the convergence of the denominator, we re-write it as $\frac{\sum \hat{\gamma}^2_k}{K}-(\frac{\sum \hat{\gamma}_k}{K})^2$.

Note that $\hat{\gamma}^2_k$ is independent but not identically distributed. When applying Kolmogorov’s strong law of large numbers, we need assumption $\sum_{k=1}^\infty \frac{Var(\hat{\gamma}_k^2)}{k^2}<\infty$. Under the assumption, we conclude that 
$$
\frac{\sum \hat{\gamma}^2_k}{K} \rightarrow \sigma^2_\gamma+\mu^2_\gamma+ \overline{\sigma^2_{uk}} 
$$

and 
$$
(\frac{\sum \hat{\gamma}_k}{K})^2 \rightarrow \mu^2_\gamma
$$ with continuous mapping theorem. Thus, we have $\frac{1}{K} \sum_{k=1}^K(\hat{\gamma}_k-\overline{\hat{\gamma}_k})^2 \rightarrow \sigma^2_\gamma + \overline{\sigma^2_{uk}}$.

For the numerator, we consider $\frac{\sum_{k=1}^K(\hat{\gamma}_k-\overline{\hat{\gamma}_k})\gamma_k}{K}$. Similarly, we find $\frac{\hat{\gamma}_k\gamma_k}{K}\rightarrow \sigma^2_\gamma + \mu^2_\gamma$ and $\frac{\overline{\hat{\gamma}_k}\gamma_k}{K}\rightarrow \mu^2_\gamma$, and thus $\frac{\sum_{k=1}^K(\hat{\gamma}_k-\overline{\hat{\gamma}_k})\gamma_k}{K} \rightarrow \sigma^2_\gamma$. 

Return to the estimator, we have 
\begin{align}
\hat{\beta}&= \frac{\frac{1}{K} \sum_{k=1}^K(\hat{\gamma}_k-\overline{\hat{\gamma}_k})[\mathbb{E}\delta_k + \gamma_k \beta_k + (\varepsilon_k+v_k)]}{\frac{1}{K} \sum_{k=1}^K(\hat{\gamma}_k-\overline{\hat{\gamma}_k})^2}\\
&\xrightarrow{p} \lambda \mathbb{E}\beta_k.
\end{align} where we use the same methods in the proof of Proposition \ref{prop:consist} and zero covariance $Cov(\gamma_k,v_k)=0$ in the assumption.

\end{proof}

\section{Two-step SIMEX estimator}\label{si:twostep}

To achieve this result, we rely on two techniques. The first is honest estimation, which has become standard in causal forests and the recent double machine learning (DML) literature. Specifically, we partition the data into two disjoint subsamples, $I_1 \cup I_2$. The first subsample $I_1$ is used to identify the $K$ subgroups, while the second subsample $I_2$ is used to estimate $(\hat{\gamma}_k,\hat{\tau}_k,\hat{\sigma}^2_{uk})$, where $\hat{\sigma}^2_{uk}$ denotes the estimated asymptotic variance of $\hat{\gamma}_k$. In practice, $k$-fold cross-fitting is recommended, as is standard in the current DML literature. When the subgroups are pre-determined based on prior studies, this sample-splitting step is unnecessary.

The second requirement is that subgroup sizes grow fast relative to the number of subgroups. Let $n_k$ denote the sample size of subgroup $k$ in $I_2$, and define $\tilde{n}_k=\inf_k n_k$. We impose the following assumption.

\begin{assumption}\label{ass:two}
$\frac{1}{\sqrt{K}}\sum_{k=1}^K \frac{1}{\tilde{n}_k} \rightarrow_p 0$ as $\tilde{n}_k \rightarrow \infty$ and $K \rightarrow \infty$.
\end{assumption}

This assumption ensures that first-stage estimation error does not contribute to the asymptotic variance in the second stage. When subgroup sizes dominate the growth of the number of subgroups, the second-stage estimation is asymptotically unaffected by first-stage noise. We can achieve this by setting a minimum number of samples in each subgroup.


As the formal proof is technically involved, to avoid excessive notation, we present the complete theorem—along with the required regularity conditions and the detailed estimation procedure—in SI~\ref{si:twostep}. Here, we focus on the key intuition underlying the result.

\begin{theorem}
Suppose $(\tau_k,\gamma_k)$ satisfies the decomposition \eqref{equ:lm3}: $\tau_k = \mathbb{E}[\delta_k]+\beta\gamma_k+\varepsilon_k$, and $(\hat{\gamma}_k,\hat{\tau}_k,\hat{\sigma}^2_{uk})$ are obtained via sample splitting. Under Assumption~\ref{ass:two} and the regularity conditions stated in SI~\ref{si:twostep}, the two-step SIMEX estimator $\hat{\beta}_{\text{SIMEX}}$ is asymptotically normally distributed:
\[
\sqrt{K}\bigl(\hat{\beta}_{\text{SIMEX}}-\beta_{\text{SIMEX}}\bigr)\;\xrightarrow{d}\; \mathcal{N}(0,V_{\text{SIMEX}}),
\]
where $V_{\text{SIMEX}}$ takes the same form as in \citet{carroll1996asymptotics}.
\end{theorem} The exact expression for $V_{\text{SIMEX}}$ is provided in SI~\ref{si:twostep}. We emphasize that the variance is unaffected by the two-step estimation procedure because Assumption~\ref{ass:two} ensures that the contribution of the first-stage estimation error is asymptotically negligible.

We introduce the nuisance parameter $\eta_k=(\gamma_k,\tau_k,\sigma^2_{u,k})$. They are estimated in the first stage using sample split. In the first half $I_1$ with sample size $n_1$, we identify subgroups. In the second half $I_2$ with size $n_2=n-n_1$, we estimate nuisance parameter, denoted by  $(\hat{\gamma}_k,\hat{\tau}_k,\hat{\sigma}^2_{u,k})$. If we do not cross-fit, we need stronger Donsker-type conditions. This may not hold if we use sophisticated machine learning algorithms. In practice, $k$-fold cross-fitting is recommended, as is standard in the current DML literature.

SIMEX Procedures:

1. Simulation Step: For each $\lambda \in \Lambda={\lambda_1,...,\lambda_M}$ and each simulation draw $b=1,...,B$,
$$
\hat{\gamma}_k^{\lambda,b}=\hat{\gamma}_k+\sqrt{\lambda}\hat{\sigma}^2_{uk}e_{kb}, \; e_{kb}\sim N(0,1)
$$ 
then compute the naive slope $\hat{\beta}^{\lambda,b}$ from regressing $\hat{\tau}_k$ on $\hat{\gamma}_k^{\lambda,b}$, and average over $b$:
$$
\hat{\beta}_*^\lambda=\frac{1}{B}\sum_{b=1}^B\hat{\beta}^{\lambda,b}.
$$

2. Extrapolation Step: fit a parametric model $\hat{\beta}_*^\lambda=G(\Gamma,\lambda)$ and set 
$$
\hat{\beta}_{SIMEX}=G(\hat{\Gamma},-1).
$$

\citet{carroll1996asymptotics} show asymptotic normality of $\hat{\beta}_{SIMEX}$ by considering asymptotic linearization in two steps. The proof depends on the estimating equations. We define $0= \sum_{k=1}^K \psi_k^{\lambda,b}(\beta,\hat{\eta}_k)$. For our case, the slope in a simple linear regression, the estimating equation could be $\psi_k^{\lambda,b}(\beta)=(\hat{\lambda}_k^{\lambda,b}-\overline{\hat{\lambda}_k^{\lambda,b}})(\hat{\tau}_k-\overline{\hat{\tau}_k}-\beta(\hat{\lambda}_k^{\lambda,b}-\overline{\hat{\lambda}_k^{\lambda,b}}))$, where $\overline{a}_k$ denotes the average $\frac{1}{K}\sum_{k=1}^K a_k$. We follow their steps and show the necessary modification because of our two-step estimations. We maintain all regularity conditions used in \citep{carroll1996asymptotics}. The proof conditional on $I_A$. I omit this condition for simplicity. 

\begin{assumption}\label{ass:twostep}

    1. First stage: $\hat{\eta}_k \rightarrow \eta_k$ as $n_k \rightarrow \infty$, $||\hat{\eta}_k-\eta_k||=O_p(n_k^{-1/2})$ for all $k$.

    2. There exists a unique $\beta(\lambda)$ solving $\mathbb{E}[\psi_k^{\lambda,b}(\beta,\eta_k)]=0$, 
    $\mathbb{E}||\dot{\psi}_k^{\lambda,b}||^2$ bounded.

    3. $\psi_k^{\lambda,b}(\beta,\eta_k)$ is twice continuously differentiable in $\beta$ and $\eta$, with derivatives uniformly dominated.
    
    4.  $\lim_k \frac{1}{K}\mathbb{E}[\frac{\partial}{\partial \beta}\psi_k^{\lambda,b}(\beta,\eta)]$ exists and nonsingular.

    5. Lindeberg condition holds for $\psi_k^{\lambda,b}(\beta,\eta_k)$.
\end{assumption}

\begin{proposition}
Suppose $\eta_k$ are estimated by honest sample splitting, and regularity assumptions \ref{ass:twostep} hold. If $$\frac{1}{\sqrt{K}}\sum_{k=1}^K \tilde{n}_k^{-1/2}\rightarrow 0,$$
then, two-step SIMEX estimator $\hat{\beta}$ is asymptotically normal distributed with the same variance (and variance estimation) in \cite{carroll1996asymptotics} as $K \rightarrow \infty$ and $\tilde{n}_k=inf_k n_k \rightarrow \infty$, 
\end{proposition}

\begin{proof}

By standard estimation equation theory, $\hat{\beta}^{\lambda,b}$ converges to the the limit $\beta^\lambda$ that solves $\mathbb{E}[\psi_k^{\lambda,b}(\beta,\eta_k)]=0$. We expand estimating equations in nuisance parameter $\eta$,
$$
\frac{1}{\sqrt{K}}\sum_{k=1}^K \psi_k^{\lambda,b}(\beta,\hat{\eta}_k)= \frac{1}{\sqrt{K}}\sum_{k=1}^K \psi_k^{\lambda,b}(\beta,\eta_k)+\frac{1}{\sqrt{K}}\sum_{k=1}^K \dot{\psi}_k^{\lambda,b}(\hat{\eta}_k-\eta_k)+\frac{1}{\sqrt{K}}\sum_{k=1}^KR_k
$$ where $\dot{\psi}_k^{\lambda,b}$ is the gradient with respect to $\eta$. The second order remainder is $R_k=O_p(n_k^{-1})$ by the regularity of the second order derivative, and
$\frac{1}{\sqrt{K}}\sum_{k=1}^KR_k=\frac{1}{\sqrt{K}}\sum_{k=1}^KO_p(n_k^{-1})$ goes to zero under assumption $\frac{1}{\sqrt{K}}\sum_{k=1}^K\tilde{n}_k^{-1/2}\rightarrow 0$. 

Now, standard expansion around $\beta^\lambda$ yields
$$
\begin{aligned}
\sqrt{K}(\hat{\beta}^{\lambda,b}-\beta^\lambda)&=-[A(\lambda)]^{-1}\frac{1}{\sqrt{K}}\sum_{k=1}^K\psi_k^{\lambda,b}(\beta^\lambda,\hat{\eta}) +o_p(1)\\
&=-[A(\lambda)]^{-1}\frac{1}{\sqrt{K}}\sum_{k=1}^K[\psi_k^{\lambda,b}(\beta,\eta_k)+\dot{\psi}_k^{\lambda,b}(\hat{\eta}_k-\eta_k)]+o_p(1)
\end{aligned}
$$ where $A(\lambda)=\lim_k \frac{1}{K}\mathbb{E}[\frac{\partial}{\partial \beta}\psi_k^{\lambda,b}(\beta,\eta)]$. Because nuisance parameters are estimated, influence function has additional term, $\dot{\psi}_k^{\lambda,b}(\hat{\eta}_k-\eta_k)$, which does not appear in the original SIMEX by \citet{carroll1996asymptotics}.

Let $\Delta_k=\dot{\psi}_k^{\lambda,b}(\hat{\eta}_k-\eta_k)$. To show $\frac{1}{\sqrt{K}}\sum_{k=1}^K \Delta_k \rightarrow_p 0$, the procedures are now standard in the double machine learning literature \citep{chernozhukov2018double}. Because cross-fitting, conditional on $I_1$, $\mathbb{E}[\Delta_k]=0$. By Cauchy–Schwarz, the conditional variance $Var(\frac{1}{\sqrt{K}}\sum_{k=1}^K\Delta_K)=\frac{1}{K}\sum_{k=1}^KVar(\Delta_k)\leq \frac{1}{K}\sum_{k=1}^K\mathbb{E}||\dot{\psi}_k^{\lambda,b}||^2 \mathbb{E}||\hat{\eta}_k-\eta_k||^2$. By assumption $\mathbb{E}||\dot{\psi}_k^{\lambda,b}||^2$ is bounded and we know $\sqrt{n_k}(\hat{\eta}_k-\eta_k)=O_p(1)$, therefore, conditional variance goes to zero if $\frac{1}{K}\sum_{k=1}^K \frac{1}{\tilde{n}_{k}}\rightarrow_p 0$. Then, by Chebyshev's inequality, we conclude $\frac{1}{K}\sum_{k=1}^KVar(\Delta_k) \rightarrow 0$.

Next, to show the influence function $\frac{1}{\sqrt{K}}\sum_{k=1}^K[\psi_k^{\lambda,b}(\beta,\eta_k)]$ converges to normal distribution, here because data comes from different distribution, we apply Lindeberg (or Lyapunov) non-iid CLT by assuming $\sup_K \mathbb{E}|\psi_K^{\lambda,b}|^{2+\delta}<\infty$ and a Lindenberg condition, we have $\frac{1}{\sqrt{K}}\sum_{k=1}^K\psi_k^{\lambda,b}(\beta^\lambda) \rightarrow N(0,\Omega(\lambda))$. 
Then, $\sqrt{K}(\hat{\beta}^{\lambda,b}-\beta^\lambda) \rightarrow N(0,A(\lambda)^{-2}\Omega(\lambda))$.

The remaining proofs exactly follows \citet{carroll1996asymptotics}. As in their section 3.2, we averages over B draws and combine to extrapolation part . Averaging over b, $\hat{\beta}^{\lambda}_*=\frac{1}{B}\sum_{b}\hat{\beta}^{\lambda}$, we have $\sqrt{K}(\hat{\beta}^{\lambda}_*-\beta^{\lambda})\rightarrow N(0,V_\lambda)$.

Finally, in the extrapolation step, we fit a parametric model $\hat{\beta}^{\lambda}_*=G(\Gamma,\lambda)$. Let $\hat{\beta}^{\Lambda}_*$ a vector of $\hat{\beta}^{\lambda}_*$ on the grid $\Lambda=(\lambda_1,...,\lambda_m)$. One now set $\hat{\beta}^{\Lambda}_*=G(\Gamma,\lambda)$ for a vector of parameters $\Gamma$ and fits $G(\Gamma,\lambda)$, $\lambda \in \Lambda$, to the elements of $\hat{\beta}^{\Lambda}_*$. Let $R(\Gamma)=\hat{\beta}^{\Lambda}_*-\{G(\Gamma,\lambda_1),...,G(\Gamma,\lambda_M) \}^T$, a typical estimator of $\hat{\Gamma}$ is obtained by minimizing $R^T(\Gamma)C^{-1}R(\Gamma)$, for some positive definite matrix $C$. This implies the estimating equation $0=s(\Gamma)C^{-1}R(\Gamma)$, where $s(\Gamma)=\{\nabla_\Gamma G(\Gamma,\lambda_1),...,\nabla_\Gamma G(\Gamma,\lambda_m)\}$. Then $\hat{\Gamma}-\Gamma$ is asymptotically normal with covariance
$$\Sigma(\Gamma)=D^{-1}(\Gamma)s(\Gamma)C^{-1}A(\lambda)^{-2}\Omega(\lambda)C^{-1}s^T(\Gamma)D^{-1}(\Gamma)$$ where $D^{-1}(\Gamma)=s(\Gamma)C^{-1}s^T(\Gamma)$.

The finial SIMEX estimator is $\hat{\beta}_{SIMEX}=G(\hat{\Gamma},-1)$. As illustrated in 3.3 \citet{carroll1996asymptotics}, the distribution is obtained through delta method,
$$\sqrt{K}(\hat{\beta}_{SIMEX}-\beta_{SIMEX})\rightarrow N(0,\nabla^T_\Gamma G(\Gamma,-1)\Sigma_\Gamma(\Gamma)\nabla_\Gamma G(\Gamma,-1))$$ 
\end{proof}

Let $\hat{\overline{\gamma}} = \frac{1}{K}\sum_{k=1}^K \hat{\gamma}_k$. The estimator of the average indirect effect is $\hat{\eta}=\frac{1}{K}\sum_{k=1}^K\hat{\beta}\hat{\gamma}_k=\hat{\beta}\hat{\overline{\gamma}}_k$. The variance of $\hat{\eta}$ can then be derived using the delta method:
$$
\begin{aligned}
Var(\sqrt{K}(\hat{\eta}-\eta)) \approx \overline{\gamma}^2 &Var(\sqrt{K}(\hat\beta_{SIMEX}-\beta_{SIMEX}))+\beta^2Var(\sqrt{K}(\hat{\overline{\gamma}}-\overline{\gamma}))\\
&+2\beta\overline{\gamma}Cov(\sqrt{K}(\hat\beta_{SIMEX}-\beta_{SIMEX}), \sqrt{K}(\hat{\overline{\gamma}}-\overline{\gamma}))
\end{aligned}
$$ The first term can be estimated consistently, as shown above. The second term can be ignored under the assumption that $\frac{1}{\sqrt{K}}\sum_{k=1}^K \tilde{n}_k^{-1/2}\rightarrow 0$. However, the third covariance term is less straightforward. It can be ignored if we further split the sample so that $\hat{\beta}{\mathrm{SIMEX}}$ is estimated using an independent subsample. There is, however, a trade-off: this additional sample splitting may increase the variance of $\hat{\beta}_{SIMEX}$. Therefore, to achieve better finite-sample performance, we instead employ an intersection–union test and confidence interval based on the test.

\section{More on Simulation and Application}\label{si:simul}

\subsection{Power Analysis}

We investigate how the power changes with the number of groups and the number of individuals within each group. In the simulation, we assume that at the beginning, there are 10 groups. Each group has population $n$. Suppose researchers can enroll another $k*n$ individuals depending on their budget ($k$ is an integer). They face a decision: whether to add $k$ more groups or to increase the group size (i.e. add $\frac{kn}{10}$ individuals to each existing groups).

The data generation process follows the same procedure as outlined in the main text. When we increase the group size, it becomes necessary to enhance the precision of the observed values. To achieve this, we employ the following approximation:

$$
se(\hat{\beta}) \approx \frac{c}{\sqrt{n}}
$$

Therefore, when adding $\frac{n}{10}$ to each group, the standard error is adjusted to $\frac{\sigma \sqrt{n}}{\sqrt{n+kn/10}}$ where $\sigma$ is the baseline standard error for $\gamma$ and $\tau$ in the simulation.

\begin{figure}
    \centering
   \includegraphics[width=1\textwidth]{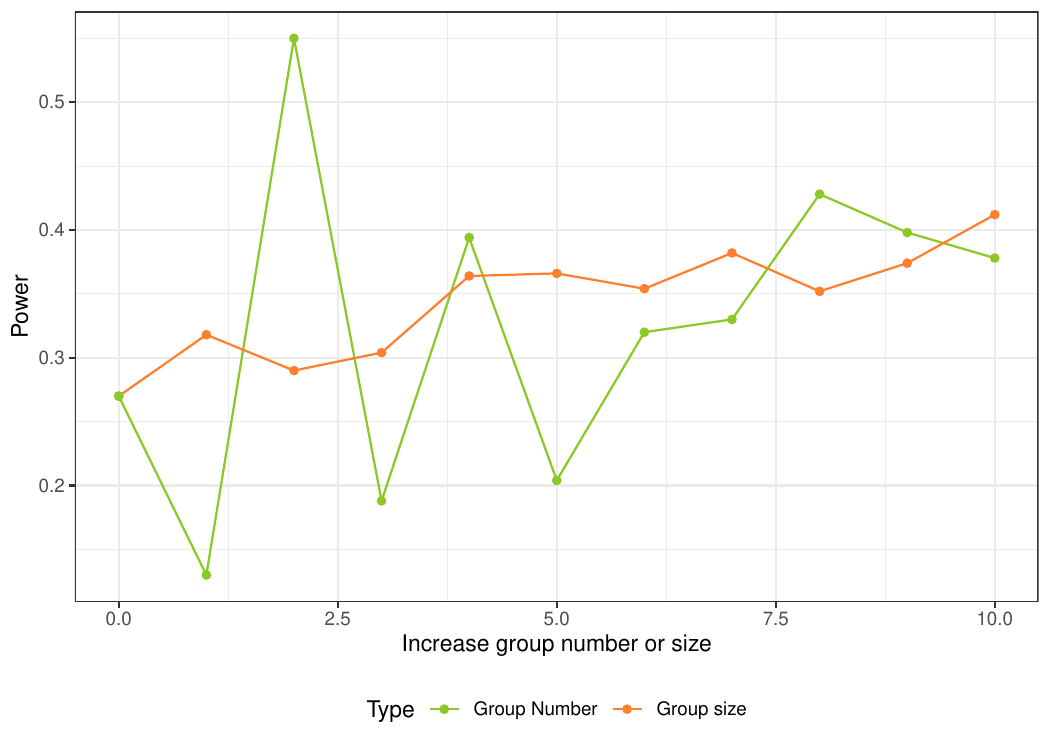}
    \caption{Power analysis: Group number and size}
    \label{fig:power}
\end{figure}

In the figure \ref{fig:power}, the vertical axis represents the power. The number on the horizontal line is $k$, which denotes $k$ more groups ( green line) or $kn/10$ more individuals in each group (orange line). The figure does not readily suggest a clear decision regarding the optimal choice between these two options. This ambiguity underscores the importance of conducting a thorough power analysis prior to the experiment to guide such decisions.

\subsection{Application II}\label{sec:app_ind}

In this section, we employ the full dataset to identify subgroups using causal trees, as depicted in Figure \ref{fig:sub_app}. This process reveals a greater number of subgroups. It is important to note that the number of detected subgroups is influenced by various factors, including the minimum number of observations required for each split. The corresponding estimates are shown in Figure \ref{fig:beta_app}. Upon examination, it is evident that the estimate of $\beta$ closely aligns with the one discussed in the main text.

\begin{figure}[!htbp]
    \centering
    \includegraphics[width=0.7\textwidth]{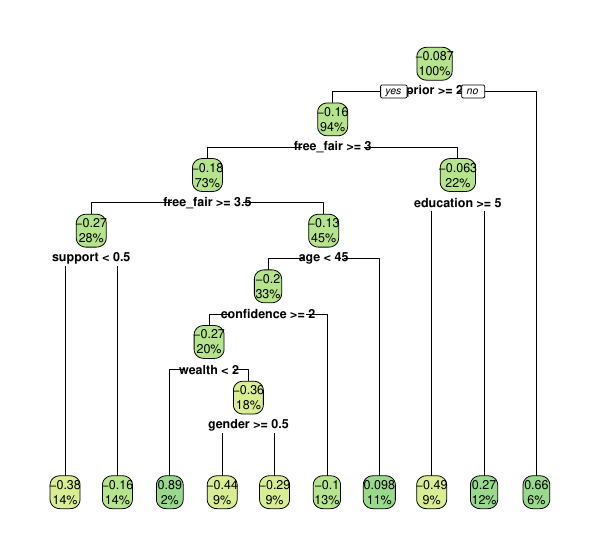}
    \caption{Heterogeneous Subgroup Design}
    \label{fig:sub_app}
\end{figure}

\begin{figure}[!htbp]
    \centering
    \includegraphics[width=0.7\textwidth]{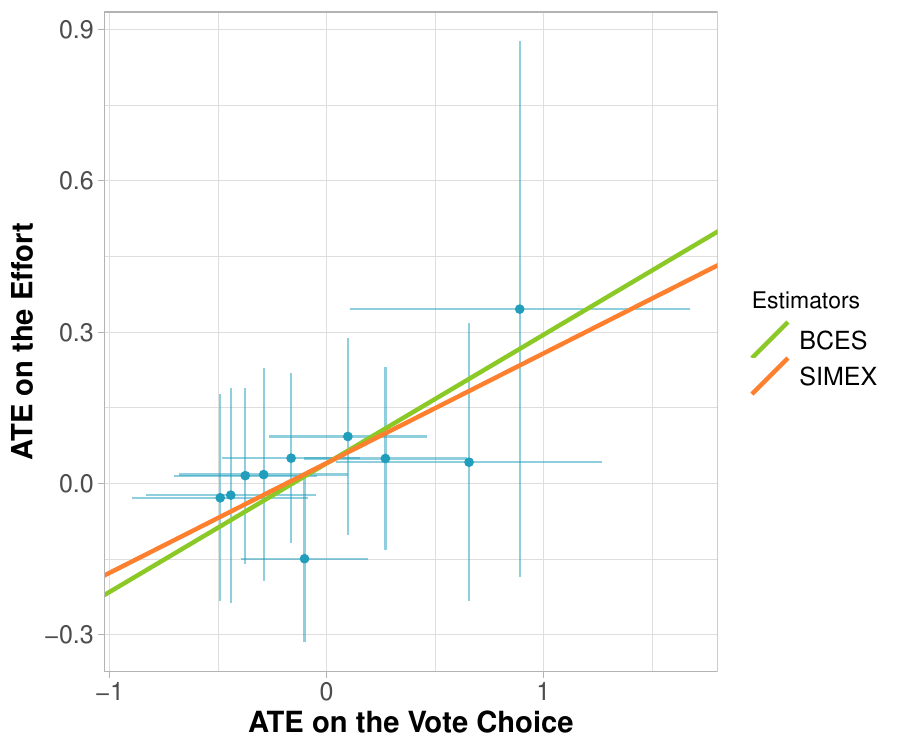}
    \caption{Heterogeneous Subgroup Design}
    \label{fig:beta_app}
\end{figure}






\end{document}